\documentclass[final,twocolumn]{IEEEtran}

\IEEEoverridecommandlockouts

\usepackage{amsmath,amssymb,amsfonts}
\usepackage{diagbox} 
\usepackage{multirow} 

\usepackage{subcaption}

\usepackage{graphicx}
\usepackage{textcomp}
\usepackage{xcolor}

\usepackage{algorithm}
\usepackage{algpseudocode}
\algnewcommand{\IfThenElse}[3]{
	\State \algorithmicif\ #1\ \algorithmicthen\ #2\ \algorithmicelse\ #3}

\usepackage{braket,amsthm,enumitem}
\usepackage[hidelinks]{hyperref}
\usepackage{tikz}

\usetikzlibrary{quantikz}


\newtheorem{definitionenv}{Definition}
\newtheorem{lemmaenv}[definitionenv]{Lemma}
\newtheorem{theoremenv}[definitionenv]{Theorem}
\newtheorem{corollaryenv}[definitionenv]{Corollary}
\newtheorem{propositionenv}[definitionenv]{Proposition}
\newtheorem{remarkenv}[definitionenv]{Remark}
\newtheorem{conjectureenv}[definitionenv]{Conjecture}
\newtheorem{exampleenv}{Example}
\newtheorem{app-lemmaenv}[section]{Lemma}

\newenvironment{definition}{\begin{definitionenv}\rm}{\end{definitionenv}}
\newenvironment{lemma}{\begin{lemmaenv}\rm}{\end{lemmaenv}}
\newenvironment{theorem}{\begin{theoremenv}\rm}{\end{theoremenv}}

\newenvironment{app-lemma}{\begin{app-lemmaenv}\rm}{\end{app-lemmaenv}}

\newcommand{\cC}{{\cal C}}
\newcommand{\cG}{{\cal G}}

\newcommand{\cS}{{\cal S}}

\newcommand{\cE}{{\cal E}}

\newcommand{\cL}{{\cal L}}

\newcommand{\UF}{\textrm{UNION-FIND}}
\newcommand{\SV}{\textrm{SYND\_VAL}}

\usepackage{accents} 		

\makeatletter 
\renewcommand*\env@matrix[1][*\c@MaxMatrixCols c]{%
	\hskip -\arraycolsep
	\let\@ifnextchar\new@ifnextchar
	\array{#1}}
\makeatother

\usepackage{xpatch}
\xpretocmd{\eqref}{Eq.~}{}{}

\usepackage[colorinlistoftodos,prependcaption]{todonotes}
\usepackage{xcolor}
\usepackage{xargs}

\newcommandx{\rednote}[2][1=]{\todo[inline,linecolor=red,backgroundcolor=red!25,bordercolor=red,#1]{#2}}
\newcommandx{\yellownote}[2][1=]{\todo[inline,linecolor=yellow,backgroundcolor=yellow!25,bordercolor=yellow,#1]{#2}}

\def\BibTeX{{\rm B\kern-.05em{\sc i\kern-.025em b}\kern-.08em
		T\kern-.1667em\lower.7ex\hbox{E}\kern-.125emX}}
\begin{document}
	
	\title{ 
		 Union-Intersection Union-Find for Decoding Depolarizing Errors in Topological Codes
	}

	\author{ Tzu-Hao Lin   and   Ching-Yi Lai
		\thanks{\footnotesize
			CYL   was supported by the National Science and Technology Council  in
			Taiwan, under Grant 113-2221-E-A49-114-MY3,  Grant 113-2119-M-A49-008-, and Grant 114-2119-M-A49-006-.

The authors are with the Institute of Communications Engineering, National Yang Ming Chiao Tung University, Hsinchu	300093, Taiwan. (email: cylai@nycu.edu.tw) 
	}
	}

	\maketitle
	
	\thispagestyle{plain}
	\pagestyle{plain}

	\begin{abstract} 
		In this paper, we introduce the Union-Intersection Union-Find (UIUF) algorithm for decoding depolarizing errors in topological codes, combining the strengths of iterative and standard Union-Find (UF) decoding.
		While iterative UF improves performance at moderate error rates, it lacks an error correction guarantee.
		To address this, we develop UIUF, which maintains the enhanced performance of iterative UF while ensuring error correction up to half the code distance. 
	 Through simulations under code capacity, phenomenological, and biased noise models, we show that UIUF significantly outperforms UF, reducing the logical error rate by over an order of magnitude (at around $10^{-5}$).
		Moreover, UIUF achieves lower logical error rates than the Minimum Weight Perfect Matching (MWPM) decoder on rotated surface codes under both the code capacity and phenomenological noise models,   while preserving efficient linear-time complexity.

	\end{abstract}


	\section{Introduction}
	
	Quantum computers have demonstrated significant advantages over classical computers in specific areas, such as Shor's factoring algorithm~\cite{Shor94} and Grover's search algorithm~\cite{Gro96}. However, the practical application of these algorithms remains challenging due to the fragility and susceptibility to errors of quantum systems. Quantum error correction (QEC) is fundamental to ensuring reliable computation in the presence of noise, which is a key challenge for realizing scalable quantum computers~\cite{Shor95,CS96,Steane96}.
 	Among QEC codes, Calderbank-Shor-Steane (CSS) type  stabilizer codes~\cite{CS96,Steane96,GotPhD} are particularly valuable as they leverage classical syndrome decoding techniques to correct quantum errors. Topological codes~\cite{Kit97_03,BK98,BM07}, such as toric and surface codes,  have gained prominence as robust solutions due to their local error correction properties and well-established fault-tolerant computation schemes~\cite{Fow+12}.

Toric and surface codes can be decoded using the Minimum Weight Perfect Matching (MWPM) decoder~\cite{DKLP02} and the Union-Find (UF) decoder~\cite{DZ20,DN21}. MWPM achieves high thresholds but is computationally intensive, with a complexity of $O(n^3)$ based on Edmonds' Blossom algorithm~\cite{Edm65}, where $n$ denotes the number of error variables (or physical qubits). In contrast, the UF decoder, regarded as an approximate implementation of MWPM~\cite{WLZ22}, provides strong decoding performance with near-linear complexity of $O(n\alpha(n))$, where $\alpha(n)$ is the inverse Ackermann function~\cite{TL84}.
Moreover, recent work has shown that UF achieves linear-time worst-case complexity~\cite{GB24}.  UF can also be extended to graphs with weighted edges to incorporate noise probabilities~\cite{HNB20,WLZ22} and is applicable to circuit-level noise.
Further improvements in UF decoding performance have been demonstrated by integrating a local predecoder based on neural networks~\cite{MPT22}.
Recently, several works have focused on reducing the computational complexity of MWPM and UF decoders, including Fusion Blossom~\cite{WZ23}, Sparse Blossom~\cite{HG25}, and Bubble Clustering~\cite{FVC25}.

As CSS codes, toric and surface codes allow for separate decoding of bit-flip ($X$) and phase-flip ($Z$) errors, which is effective under an independent $X/Z$ error model. However, for general depolarizing errors, where $Y$ errors emerge as correlated combinations of $X$ and $Z$ errors, standard MWPM and UF decoders fail to capture these $X/Z$ correlations, often leading to suboptimal decoding performance~\cite{DT14}.
In this paper, we address the problem of decoding depolarizing errors in the code capacity noise model with perfect syndrome measurements and the phenomenological noise model with measurement errors.

Belief Propagation (BP)~\cite{Gal62,Pea88}-based decoders with quaternary alphabets inherently utilize the error probability distributions~\cite{KL20,KL22}.
Postprocessing techniques, such as ordered statistics decoding~\cite{PK19,RWBC20,KKL23}, can enhance BP decoding performance for topological codes but incur a  computational cost of $O(n^3)$.
Moreover,  these methods lack a distance guarantee on their error correction capability, especially in the low-error-rate regime where
achieving a logical error rate below $10^{-10}$ is crucial for fault-tolerant quantum computation.

 The MWPM decoder can be adapted for correlated errors by reweighting the edges of decoding graphs~\cite{Fow13,DT14,PF23}.
 It has been shown that reweighting the decoding graph $G_Z$ based on the decoded output from $G_X$ can improve the performance of MWPM decoding~\cite{Fow13,DT14}.
Paler and Fowler further proposed a pipelined approach with simplified reweighting to enhance the efficiency of  the reweighted  MWPM decoding~\cite{PF23}.
 Yuan and Lu~\cite{YL22} proposed an iterative refinement of the reweighting method introduced in~\cite{Fow13,DT14} and showed that the weight of the output correction does not increase with further iterations.
 Moreover, iterated MWPM has been shown to benefit from graph reweighting for asymmetric errors~\cite{iOMFC23}.

In this paper, we extend this line of research to the UF decoder, which is a promising candidate for near-term quantum devices~\cite{25, 26, 27, 28, HNB20}.
 We begin by discussing how $X/Z$ correlations can be leveraged in decoding strategies that treat $X$ and $Z$ errors separately. We then propose an iterative version of the UF decoder, referred to as IRUF, where the outputs from the $X$ and $Z$ UF decodings are used as erasures to iteratively refine subsequent error correction. This approach is based on the observation that qubits identified as erasures are more likely to contain correlated $Y$ errors, thereby enhancing performance in high-error regimes by improving the likelihood of accurate error correction.

However, we note a key limitation of this iterative decoder: unlike the standard UF decoder~\cite{DN21}, it does not guarantee the correction of errors with weights up to half the code distance. This limitation arises from the degenerate nature of quantum codes, where the decoder may output a degenerate error equivalent to the actual error. Such errors cannot reliably serve as erasures for subsequent decoding rounds.

 To address this issue, we introduce the Union-Intersection Union-Find (UIUF) decoding algorithm. In the UF algorithm, an erasure cluster is first expanded to encompass a valid error matching the given syndrome, followed by a peeling decoder. 
 To improve performance, UIUF refines the erasure selection by taking the intersection of the $X$ and $Z$ clusters as the erasure set for two subsequent UF decodings of $X$ and $Z$ errors. This approach ensures that only reliable erasures are used for correction. Furthermore, we prove that UIUF guarantees error correction for errors up to half the code distance. Notably, the time complexity of UIUF remains $O(n)$.

Through numerical simulations,    we demonstrate that  UIUF   outperforms   UF   across a variety of topological codes, including toric, rotated toric, surface, and rotated surface codes in terms of logical error rate performance. 
Simulations are performed under the code capacity, phenomenological, and biased noise models.
Specifically, at a logical error rate of approximately $10^{-5}$ for UF on toric or surface codes around $d=10$, UIUF improves performance by over an order of magnitude, with even greater gains at higher distances.
 This is due to the ability of UIUF to correct more low-weight errors by exploiting $X/Z$ correlations.

For threshold analysis, UIUF with  weighted growth achieves thresholds of 15.5\%–15.6\% under the code capacity noise model. This matches the threshold of 15.5\% reported for MWPM in~\cite{WFSH10}. 
In the phenomenological noise model, UIUF improves upon UF by increasing the threshold value from 3.40\% to 3.47\% for rotated surface codes. Although UIUF’s threshold in the phenomenological noise model (3.47\%) is lower than MWPM’s threshold of 3.92\%, we demonstrate that UIUF achieves lower logical error rates than MWPM in both the code capacity and phenomenological noise models.

	 This paper is organized as follows.
	 In the next section, we introduce the background knowledge of quantum codes, including topological codes, and outlines the basic algorithms for the UF decoder. In Section~\ref{sec:IRUF}, we discuss $X/Z$ error correlations and present the IRUF decoder. Then we  introduced the UIUF decoder and provide a proof of its decoding distance guarantee in Section~\ref{sec:UIUF}. Simulation results are presented in Section~\ref{sec:sim}. Finally, we conclude with a discussion of potential future research directions.

	\section{ Quantum stabilizer codes and related decoding problems} \label{sec:basics}
 
  We consider the fundamental unit of quantum information encoded in qubits, with the computational basis ${\ket{0}, \ket{1}} \in \mathbb{C}^2$.
  
  Let the single-qubit Pauli operators be defined as follows:
   $I \triangleq \begin{bmatrix} 1 & 0 \\ 0 & 1 \end{bmatrix}$, $X \triangleq \begin{bmatrix} 0 & 1 \\ 1 & 0 \end{bmatrix}$, $Y \triangleq \begin{bmatrix} 0 & -i \\ i & 0 \end{bmatrix}$, $Z \triangleq \begin{bmatrix} 1 & 0 \\ 0 & -1 \end{bmatrix}$. Note that $X$, $Y$, and $Z$ anticommute with each other.
  Denote the $n$-fold Pauli group by $\langle \mathcal{G}_n, \cdot \rangle  \triangleq \{ c P_1 \otimes P_2 \otimes ... \otimes P_n \vert c \in \{ \pm 1, \pm i \}, P_j \in \{ I, X, Y, Z \} \}$. 

		  Denote $X_i$ as a Pauli operator that applies  $X$ to the $i$-th qubit while acting as the identity on all other qubits.
		  $Y_i$ and $Z_i$ are  similarly defined. 
		  The weight of an $n$-fold Pauli operator is the number of qubits on which it acts non-trivially (i.e., with a non-identity Pauli component),
		  	 while its support is the set of qubit positions where these nontrivial operations occur.

Let $\cS$ be an abelian subgroup of $\cG_n$ that does not contain the minus identity. 
 	A stabilizer code $C(\cS)$  is defined as the subspace  stabilized by the stabilizer group $\cS$:    	 
$$		C(\cS) = \{ \ket{\psi}\in \mathbb{C}^{2^n}: g  \ket{\psi} =\ket{\psi}, \forall g \in \cS \}.$$
The code   encodes $k$ logical qubits into $n$ physical qubits 
if    $\cS$ has $n-k$ independent generators \cite{GotPhD}.
	The elements of $\cS$ are referred to as  stabilizers. 
	
	A Pauli error $E \in \mathcal{G}_n$ acting on a codeword in $C(\mathcal{S})$ can be detected by measuring the stabilizer generators and checking the commutation relations between $E$ and the stabilizers. 
	Let  $N(\cS) = \{ P\in\cG_n\vert  P\cS = \cS P \}$ denote the normalizer group of $\cS$,
	which contains Pauli operators that commute with $\cS$. 
	The code has  minimum distance $d$ if any Pauli operator in $N(\cS)\setminus \mathcal{S}$, with weight fewer than $d$, is detectable. This implies that any error with weight up to $\lfloor \frac{d - 1}{2} \rfloor$ is correctable.
	The stabilizer code $C(\mathcal{S})$ is referred to as an [[$n$, $k$, $d$]] code.
 
The error syndrome $\sigma \in \{0, 1\}^{m}$ of an error $E$, where $m \geq n-k$, is defined by its commutation relations with a set of $m$ stabilizers $g_1, \dots, g_m$. Specifically, $\sigma_i = 1$ if $E$ anticommutes with $g_i$, and $\sigma_i = 0$ otherwise.

 \subsection{Topological Codes}
 
 We consider \textit{Calderbank-Shor-Steane} (CSS) codes~\cite{CS96,Steane96}, which have stabilizer generators composed entirely of either $X$ or $Z$ operators.
Specifically, topological codes form a class of CSS stabilizer codes characterized by a geometric arrangement of qubits and local stabilizers for error syndrome measurements~\cite{Kit03,BK98,FM01,DKLP02,RH07,BM07}. 
 In this article, we focus on surface codes and toric codes defined on two-dimensional square lattices, which are among the leading candidates for quantum computing architectures.

   Kitaev's $[[d^2 + (d-1)^2, 1, d]]$ surface codes~\cite{Kit03,BK98} are illustrated in Figure\,\ref{fig:surface_d4}~(a),    showing the layout graph for a code with distance $d = 4$. 
 The circles indicate the locations of data qubits, which correspond to potential error variables and are referred to as \textit{variable nodes}. 	
 
 Error syndromes in surface or toric codes are obtained by measuring each local $X$ or $Z$ stabilizer, which typically involves the surrounding two, three, or four data qubits, depending on the code's geometry.
 These stabilizer measurements are depicted by red or blue  boxes in the layout graph and are referred to as \textit{$X$- or $Z$-check nodes}, respectively.
 
 Toric codes resemble surface codes but are defined on a torus without boundaries. The rotated versions of toric and surface codes offer an alternative construction method that reduces overhead by rotating the original lattice by 45 degrees~\cite{BM07}. The layout graph of rotated toric codes with $d=4$ is shown in Figure\,\ref{fig:surface_d4}~(b). 
  The parameters of 
 toric and surface codes are   summarized in Table\,\ref{tb:2d_topological}.

 	\begin{figure}[h] 
 	\centering
 	\begin{subfigure} {0.23\textwidth}
 		\centering
 		\includegraphics[height=!,width=0.99\columnwidth, keepaspectratio=true]{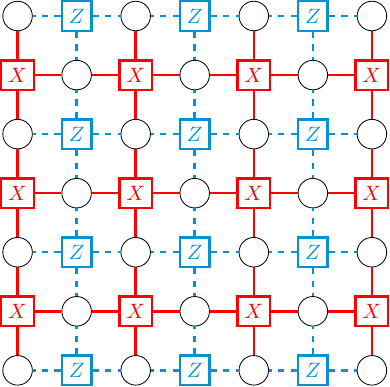}
 		\caption{surface code}
 	\end{subfigure}\hfill\begin{subfigure} {0.23\textwidth}
 		\includegraphics[height=!,width=0.99\columnwidth, keepaspectratio=true]{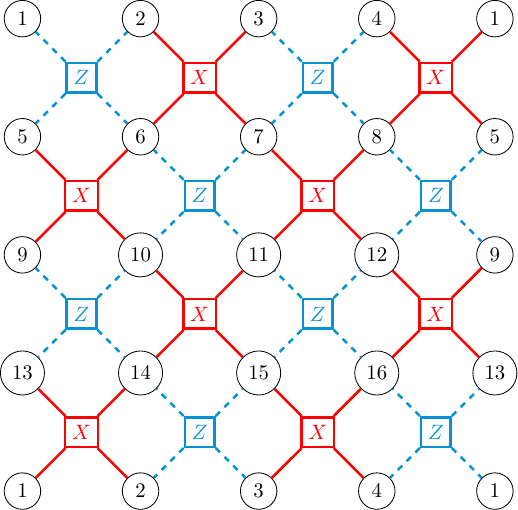}
 		\caption{rotated toric code}
 	\end{subfigure}
 	\caption{(a) Layout graph of the   surface code for $d = 4$.
 		(b) Layout graph of the   rotated toric code for $d = 4$. 
 		Circles with the same label represent the same qubit.}
 	\label{fig:surface_d4}
 \end{figure}

 \begin{table}
 	\centering
 	\footnotesize
 	\begin{tabular}{|c|c|c|}
 		\hline
 		code family& $[[n,k,d]]$& note\\
 		\hline
 		toric codes&  [[$2d^2$, 2, $d$]]& \\
 		rotated toric codes& $ [[d^2,2,d]]$& $d$ even\\
 		surface codes&$[[d^2 + (d-1)^2, 1, d]]$& \\
 		rotated surface codes &$ [[d^2,1,d]]$ & $d$ odd  \\
 		\hline
 	\end{tabular}
 	\caption{Parameters of topological codes.}\label{tb:2d_topological}
 \end{table}

For a surface code, the logical $Z$ operator is a product of $Z$ operators along a boundary-to-boundary path, while the logical $X$ operator is a product of $X$ operators along an orthogonal path.
For a toric code, logical $Z$ and $X$ operators are similarly defined along paths wrapping around the two nontrivial loops on the torus.   
 The minimum weight of a logical operator corresponds to the lattice length, which defines the minimum distance of the code.

	\subsection{Code capacity  noise model}
	We first consider the code capacity noise model with perfect syndrome extraction. We assume that each qubit experiences independent depolarizing Pauli errors at a physical error rate of $\epsilon$, where a qubit undergoes an $X$, $Y$, or $Z$ error with   probability $\epsilon/3$.

	The stabilizer code decoding problem is defined as follows:
	\begin{definition}[Syndrome Decoding Problem] \label{def:decoding}
		Given a set of $m$ stabilizers  of a stabilizer group $\mathcal{S}$ and an error syndrome $\sigma \in \{0, 1\}^m$ corresponding to an error $E \in \{I, X, Y, Z\}^{\otimes n}$, the goal of syndrome decoding is to find an error estimate $\hat{E} \in \{I, X, Y, Z\}^{\otimes n}$ such that $\hat{E} E \in \mathcal{S}$, up to a global phase.
	\end{definition}
	
	The error estimate is then applied as a correction operator for error recovery.
	A logical error occurs if $\hat{E} E$ is not a stabilizer in $\mathcal{S}$ up to a global phase.

	In addition to depolarizing Pauli errors, we also consider erasure errors~\cite{GBP97}, where a qubit is erased and replaced by a state $\ket{e}$, which is orthogonal to both $\ket{0}$ and $\ket{1}$. 
	We can assume that an erased qubit suffers uniform Pauli errors $I,X,Y,Z$ with equal probability.
	It is assumed that erasure events are detectable, and the location of the erased qubit is known.
	The erasure decoding problem is described as follows.
	\begin{definition}[Erasure Decoding Problem] \label{def:erasure_decoding}
		Given a set of $m$ stabilizers of a stabilizer group $\mathcal{S}$, an erasure set ${\cal E}\subset\{1,\dots,n\}$, and an error syndrome $\sigma \in \{0, 1\}^m$ corresponding to an error $E \in \{I, X, Y, Z\}^{\otimes n}$
		with support contained in $\cE$, the goal of erasure decoding is to find an error estimate $\hat{E} \in \{I, X, Y, Z\}^{\otimes n}$ 	with support contained in $\cE$ such that $\hat{E} E \in \mathcal{S}$, up to a global phase.
	\end{definition}

	It is also possible to address a mixed decoding problem where some qubits experience erasures, while others are subject to depolarizing errors.

	For CSS codes, the error syndromes of $X$ and $Z$ stabilizers can be processed separately to decode $Z$ and $X$ errors. Each variable node in the layout graph of a toric or surface code  represents potential errors on the data qubits and connects to at most two $X$-check nodes 
	and  two $Z$-check nodes.
	Therefore, the layout graph can be deomposed into two \textit{decoding graphs}, $G_X$ and $G_Z$.  The decoding graph $G_X$ is constructed with a vertex set of $X$-check nodes, where two $X$-check nodes are connected if they share a common variable node in the layout graph. The decoding graph $G_Z$ is defined similarly.
	
	The decoding graphs for the surface code and rotated toric code in Figure\,\ref{fig:surface_d4} are shown in Figure\,\ref{fig:decoding_d4}.
	These graphs together resemble the original layout graphs, with intersections of solid and dashed lines corresponding to data qubit locations.

For surface codes, variable nodes connected to only one check node in the layout graph are appended with virtual check nodes, as indicated by dotted boxes   in  Figure\,\ref{fig:decoding_d4}~(a).

	\begin{figure}[h] 
	\centering
			\begin{subfigure}{0.23\textwidth}
				\centering
		\includegraphics[height=!,width=0.99\columnwidth, keepaspectratio=true]{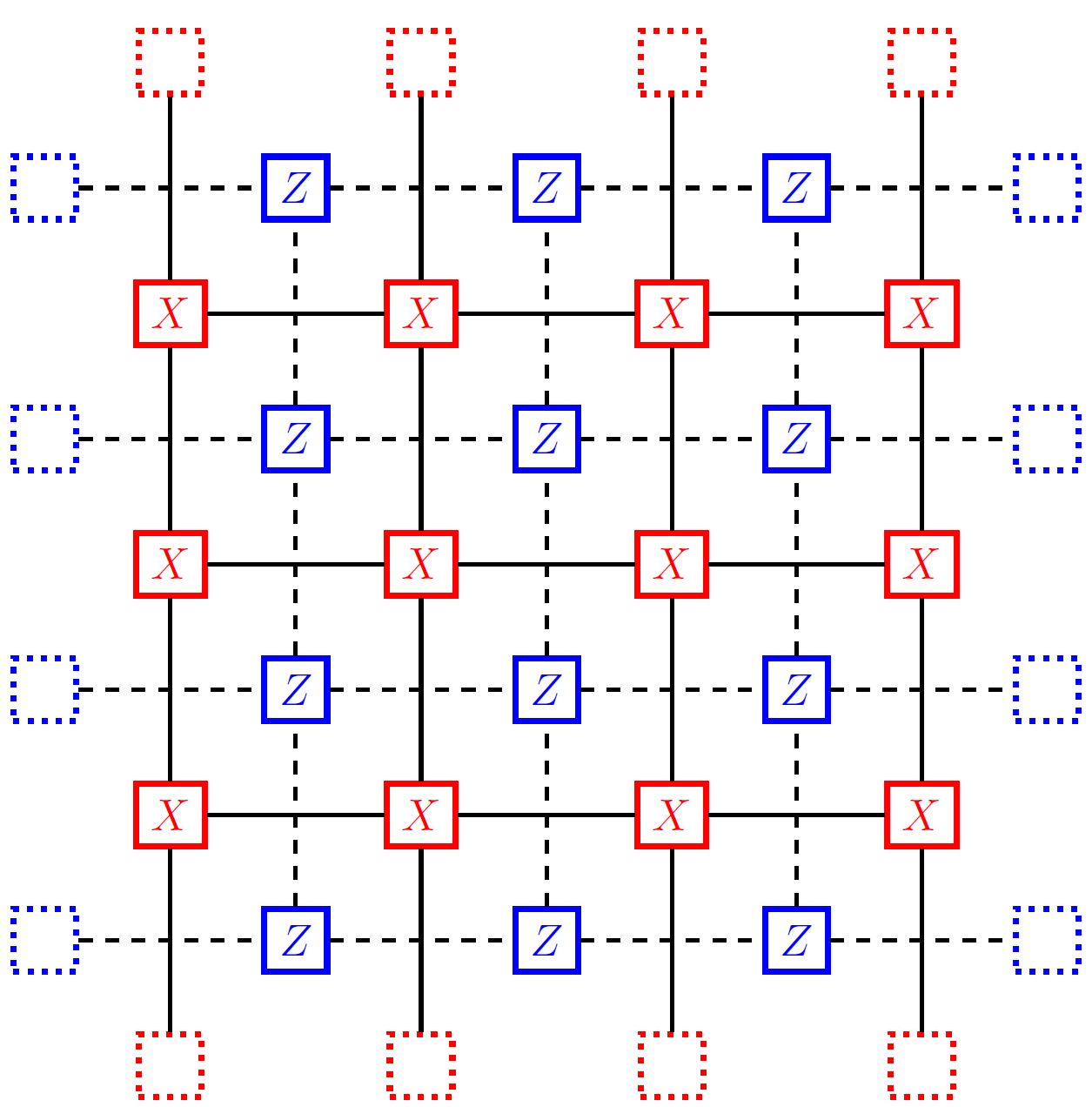}
		\caption{decoding graphs for the $[[25,1,4]]$ surface code}
	\end{subfigure}\hfill
	\begin{subfigure}{0.23\textwidth}
		\centering
		\includegraphics[height=!,width=0.99\columnwidth, keepaspectratio=true]{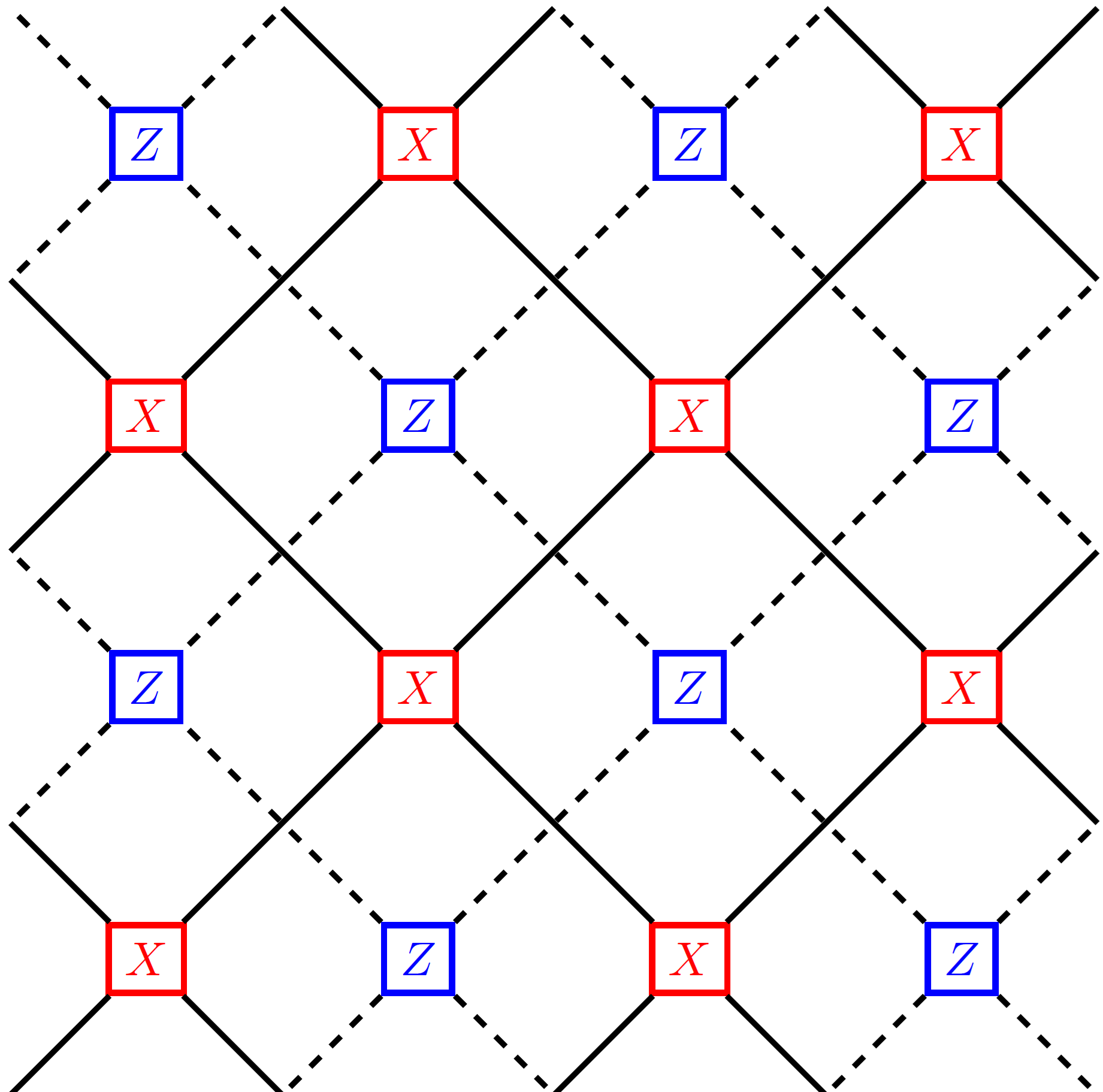}
		\caption{decoding graphs for the $[[16,2,4]] $ rotated toric code}
	\end{subfigure}
	\caption {The decoding graph $G_X$ consists of the red check nodes and solid lines, whereas the decoding graph $G_Z$ consists of the blue check nodes and dashed lines.}
	\label{fig:decoding_d4}
\end{figure}

  \subsection{Peeling decoder}

Delfosse and Z\'emor  proposed a  Peeling decoder  for surface and toric codes that corrects erasure errors~\cite{DZ20}.
\begin{lemma}[Peeling Decoder \cite{DZ20}]  \label{lemma:peeling}
	There exists a linear-time   decoder capable of correcting up to 
	$r$ erasure errors on  a toric or surface   code of distance $d$, provided   $r<d$.

\end{lemma}

For an erasure decoding problem, the set of erasures $\cE$ corresponds to two edge subsets in the decoding graphs $G_X$ and $G_Z$, effectively transforming the problem into two binary erasure decoding problems.
Recall that the intersection of a solid line in $G_X$ and a dashed line in $G_Z$ represents the location of a data qubit. When this qubit is erased, the corresponding edges in  $G_X$ and $G_Z$ will be marked as erasures. 
Thus, we will refer to these edges as belonging to the erasure set $\cE$.

A  check node is said to be \textit{nontrivial} if its syndrome is nonzero.
Let $\Sigma_X$  denote the set  of nontrivial $X$-check nodes. The goal of a decoder is to find the minimum number of edges in $\cE$ that match $\Sigma_X$ in $G_X$. To achieve this, the peeling decoder performs two steps: 
\begin{enumerate}
	\item Spanning forest growth: It first constructs a spanning forest $T_X \subset G_X$, which is a maximum subset of  $\cE$ that span the nodes in $\Sigma_X$ without  cycles. This step can  be completed in linear time.
	
\item	Peeling: Next, the peeling decoder removes edges in $T_X$ that are connected to the forest at only one endpoint. This process can be performed by traversing the forest in linear time. This process identifies a correction path between the syndrome pairs within each spanning tree, corresponding to the desired $Z$ correction.
\end{enumerate}

Similarly, let $\Sigma_Z$ denote the set  of nontrivial $Z$-check nodes. A spanning forest $T_Z \subset G_Z$  can be constructed, and the peeling step is applied to obtain  the $X$ correction.

	\subsection{Union-Find decoder} \label{chapter:UF}
	
	The UF decoder can jointly decode erasures and Pauli errors by reducing the Pauli error decoding problem to an erasure decoding problem~\cite{DN21}. Once this reduction is done, the peeling decoder is used to determine the correction operator. 
	
	The first step of the UF decoder, known as \textit{syndrome validation}, groups nontrivial check nodes into clusters, where a cluster is a connected component of check nodes influenced by errors. The decoder then expands these clusters to ensure they contain errors consistent with the observed syndromes.
		More specifically, given a set of erasures $\cE$ and a set of nontrivial $X$-check nodes $\Sigma_X$ corresponding to these erasures and any additional $Z$ errors, a cluster refers to a connected component in $G_X$ that includes edges from $\cE$ and nodes from $\Sigma_X$. Initially, the clusters consist of either individual $X$-check nodes or connected edges in $\cE$ possibly accompanied by some nodes from $\Sigma_X$. A cluster is said to be valid if its edges are consistent with a potential $Z$ error that matches its nontrivial  check nodes. Specifically, a cluster is valid if it either contains an even number of nodes from $\Sigma_X$ or connects to  a virtual node on the boundary. This is because a $Z$ error in a toric or surface code can trigger either two or one $X$-check nodes.

	Invalid clusters then expand outward by one half-edge. When two or more invalid clusters become connected, they are fused into a single valid cluster. This process repeats until all remaining clusters are valid. The peeling decoder then treats these clusters as an erasure set, resolving errors based on the syndrome information.

	\begin{definition}
		Let  $\SV(G,\Sigma)$ denote the function of syndrome validation, which takes as input a decoding graph $G$ and a set of nontrivial check nodes $\Sigma$, and returns a collection of valid clusters.
	\end{definition}

	\begin{definition}
		Let $\UF(G, \cE, \Sigma)$ denote the function of the UF decoder, which takes as input a decoding graph $G$, an erasure set $\cE$, and a set of nontrivial check nodes $\Sigma$, and returns a correction set $\cC$, which indicates the qubits that need to be corrected.
	\end{definition}

		\begin{figure}[htbp]  
		\centering  
		\begin{subfigure}{0.15\textwidth}
			\centering
			\includegraphics[height=!,width=0.99\columnwidth, keepaspectratio=true]{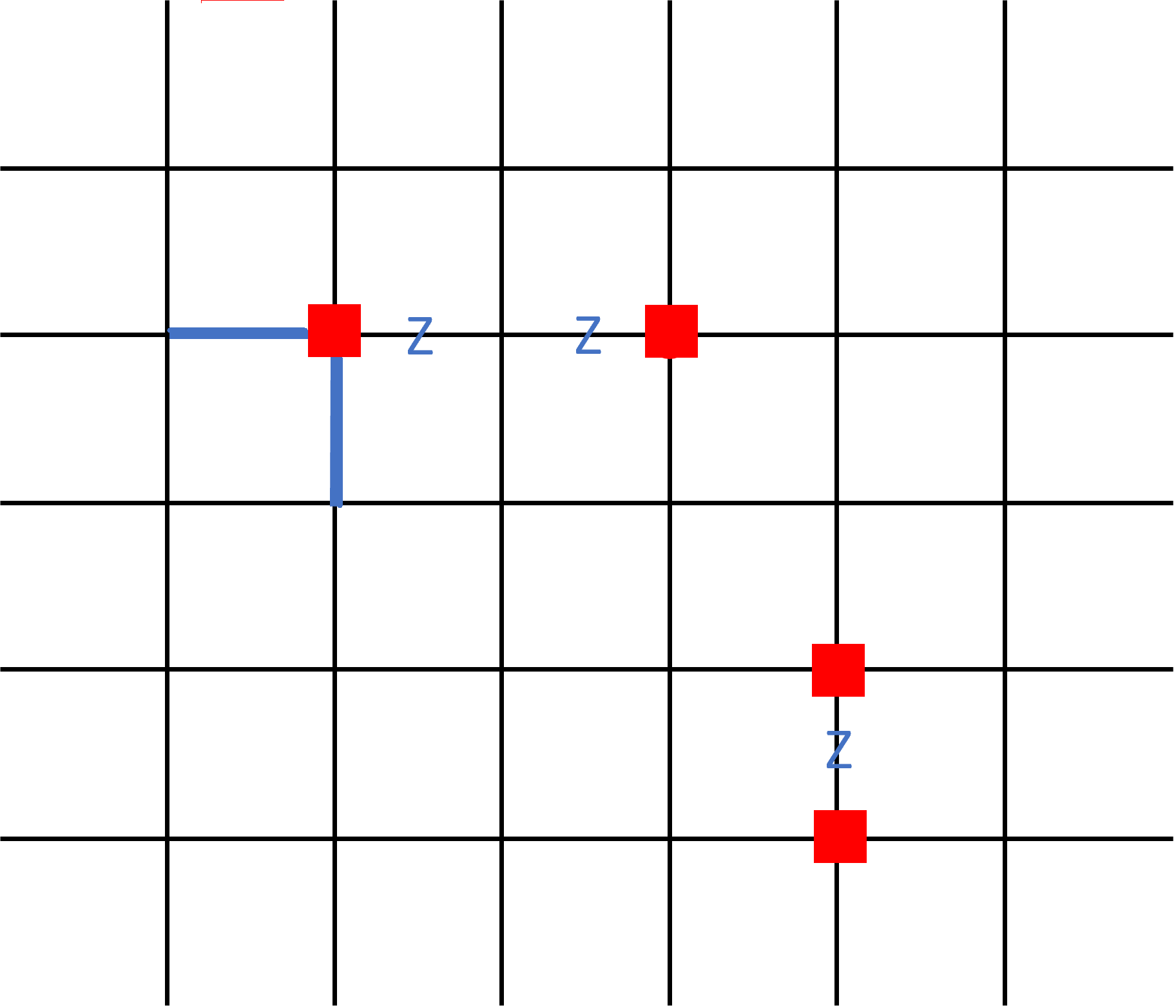} 
			\caption{\scriptsize Errors and erasures }
			\label{fig:3.6(a)}
		\end{subfigure}
	\hfil
		\begin{subfigure}{0.15\textwidth}
			\centering
			\includegraphics[height=!,width=0.99\columnwidth, keepaspectratio=true]{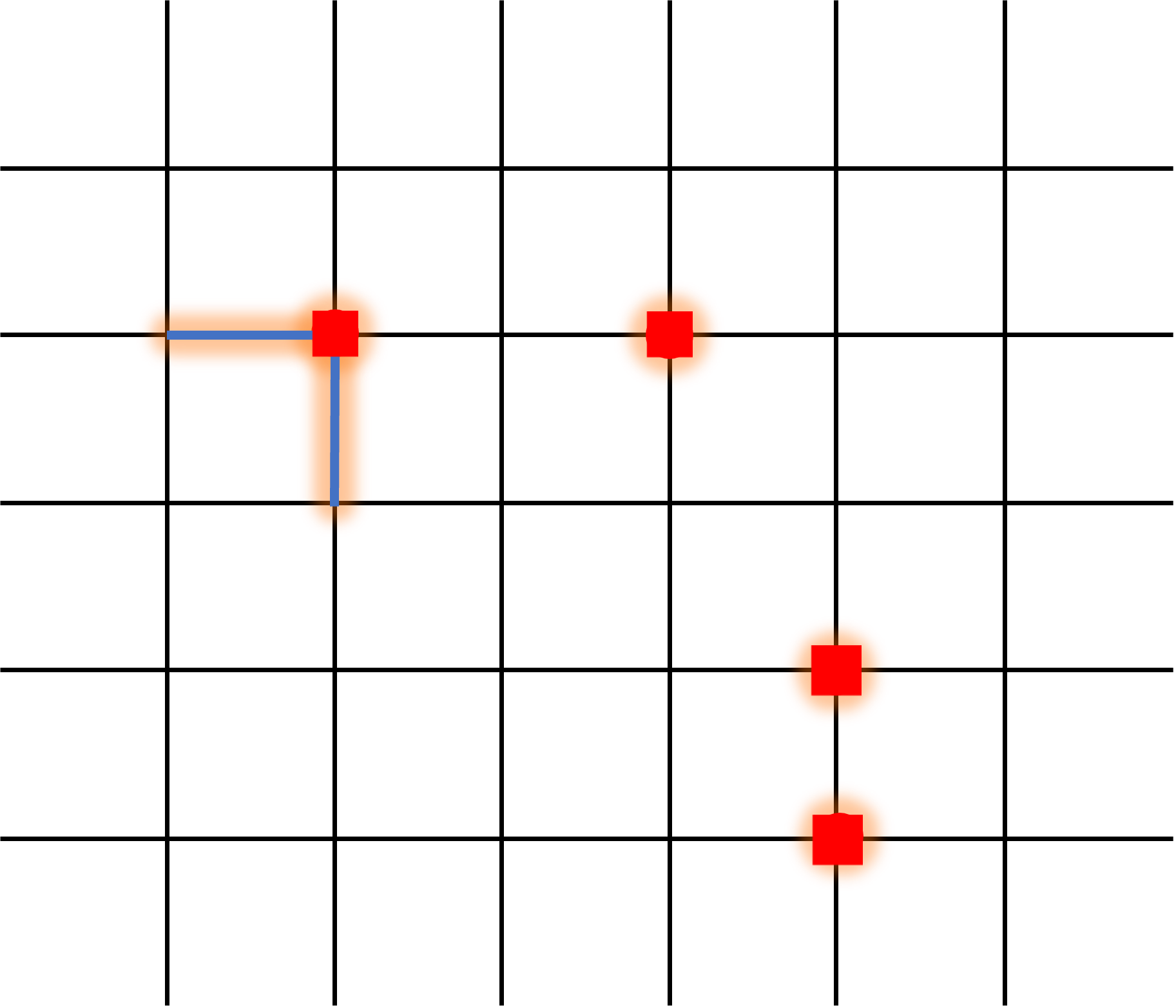} 
			\caption{\scriptsize Cluster initialization}
			\label{fig:3.6(b)}
		\end{subfigure}
	\hfil
		\begin{subfigure}{0.15\textwidth}
			\centering
			\includegraphics[height=!,width=0.99\columnwidth, keepaspectratio=true]{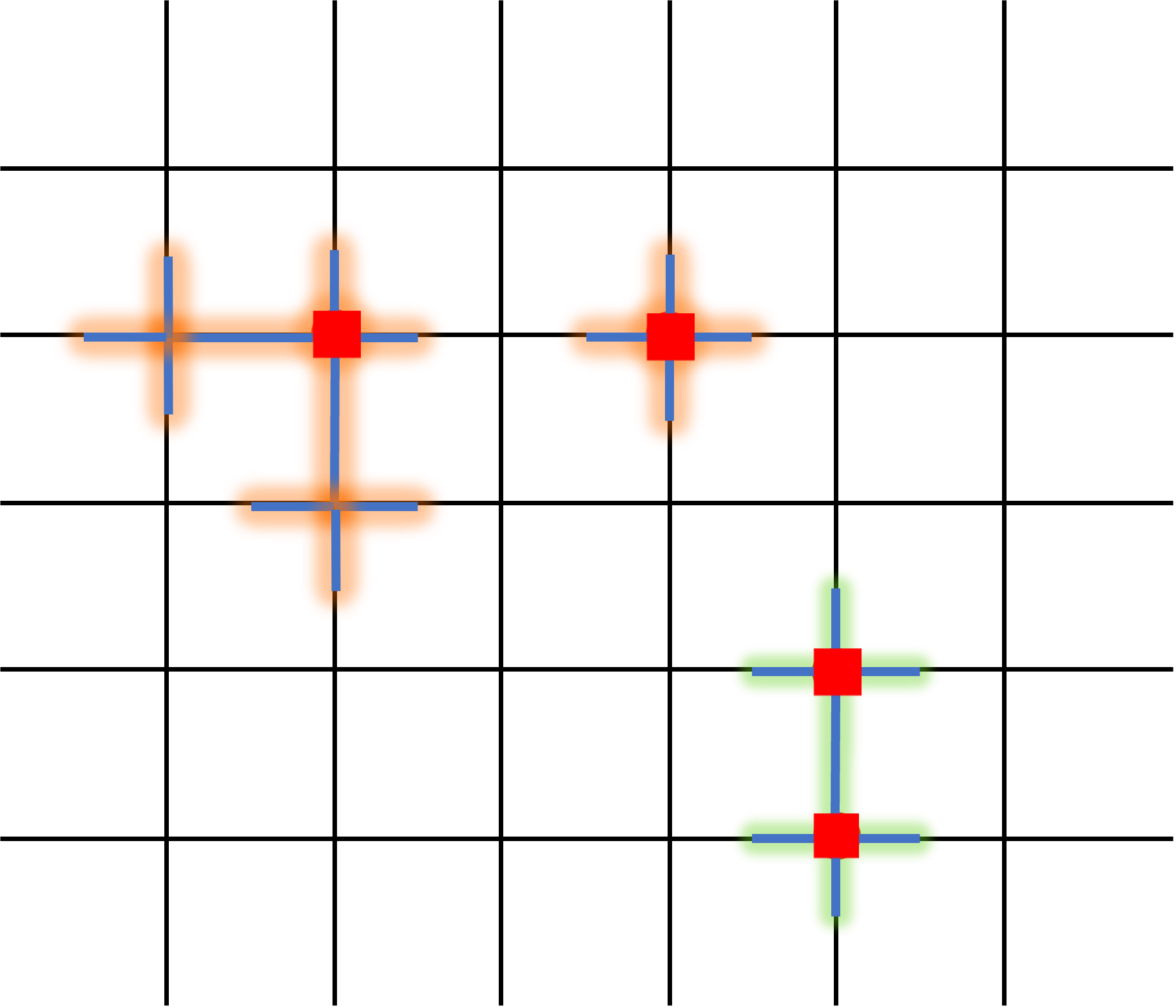} 
			\caption{\scriptsize Growth and fusion}
			\label{fig:3.6(d)}
		\end{subfigure}

		\begin{subfigure}{0.15\textwidth}
			\centering
			\includegraphics[height=!,width=0.99\columnwidth, keepaspectratio=true]{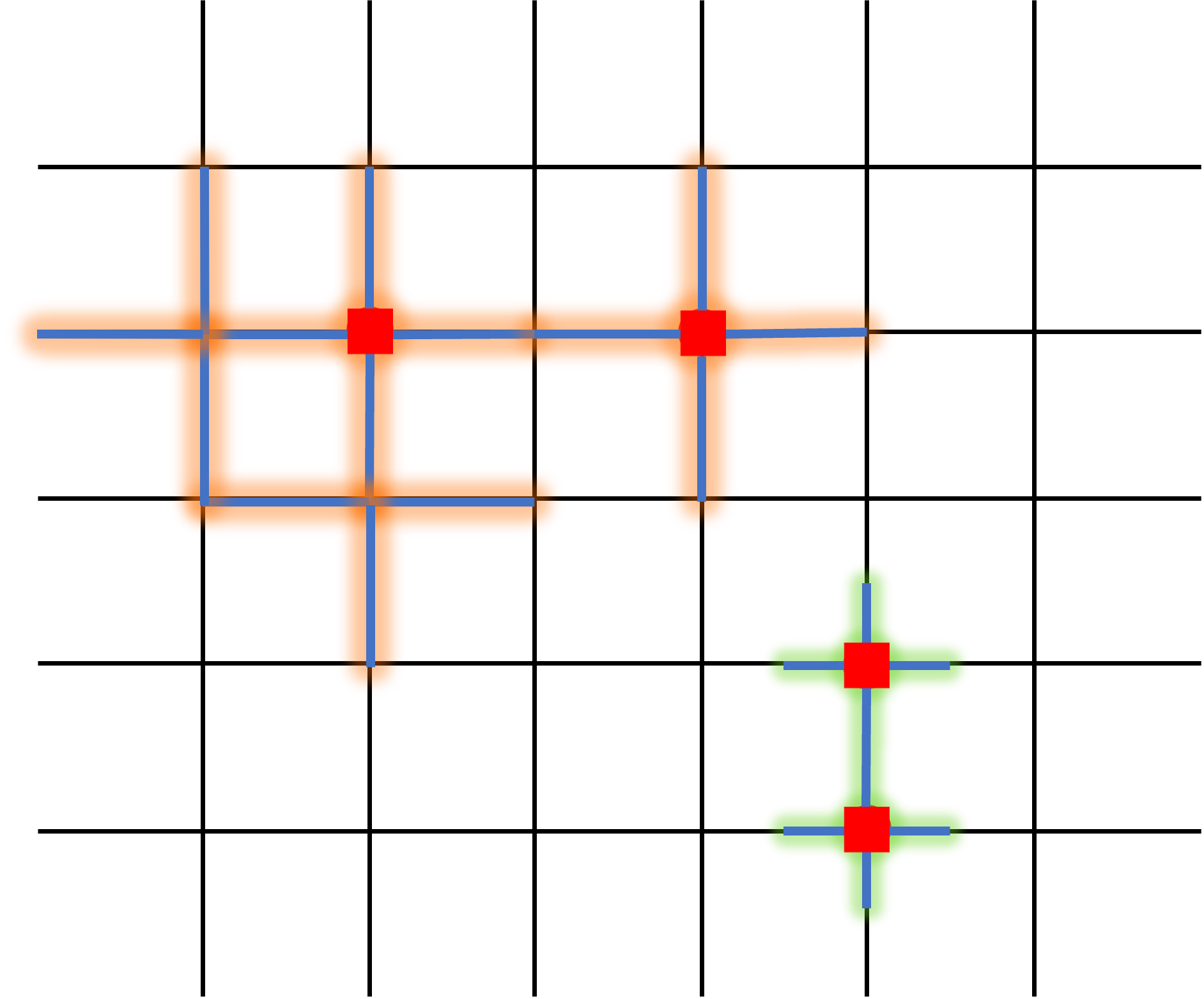} 
			\caption{\scriptsize Growth and fusion}
			\label{fig:3.6(e)}
		\end{subfigure}
 		\hfil
 		\begin{subfigure}{0.15\textwidth}
			\centering
			\includegraphics[height=!,width=0.99\columnwidth, keepaspectratio=true]{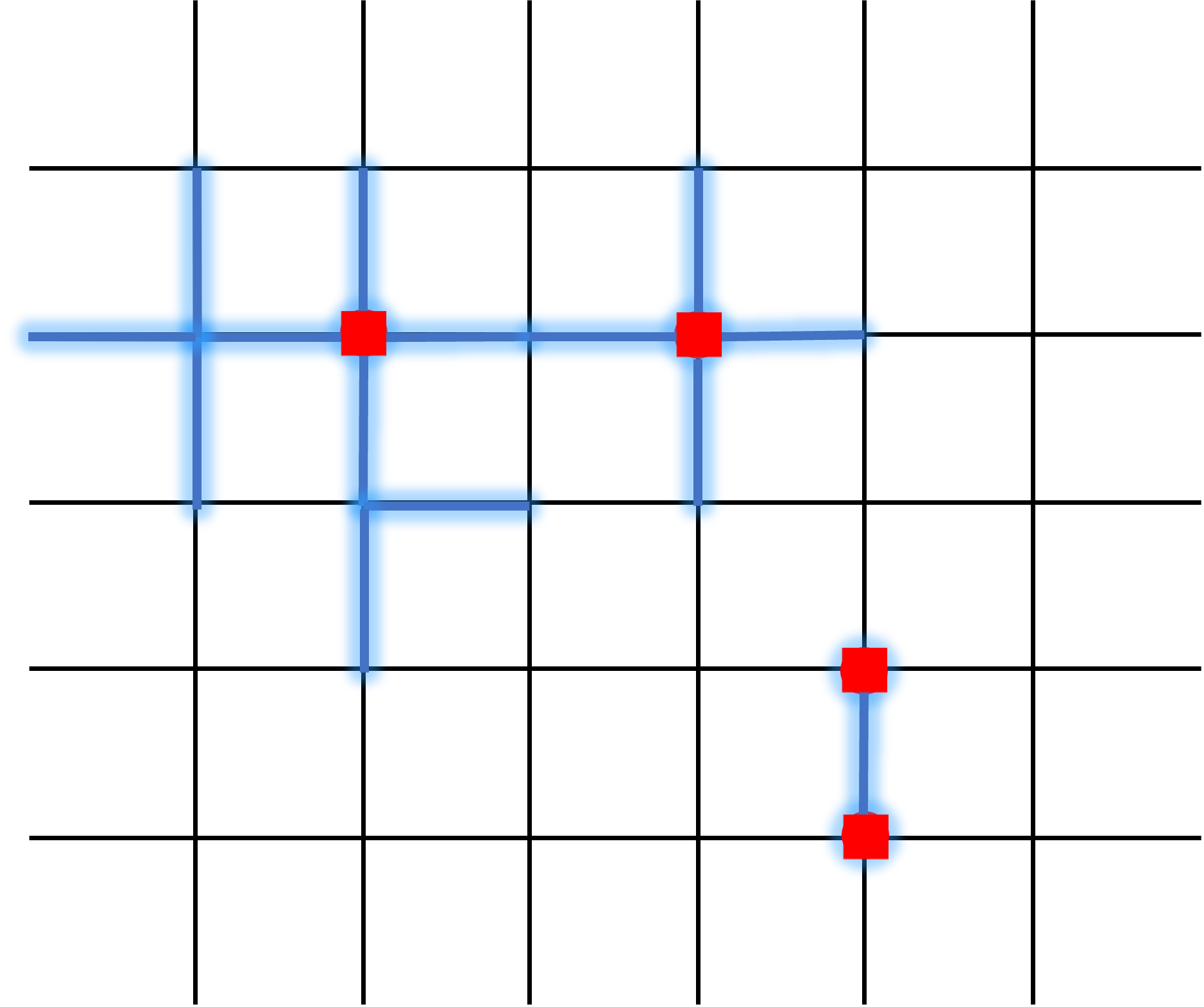} 
			\caption{\scriptsize Spanning forest}
			\label{fig:3.6(g)}
		\end{subfigure}
	\hfil
		\begin{subfigure}{0.15\textwidth}
			\centering
			\includegraphics[height=!,width=0.99\columnwidth, keepaspectratio=true]{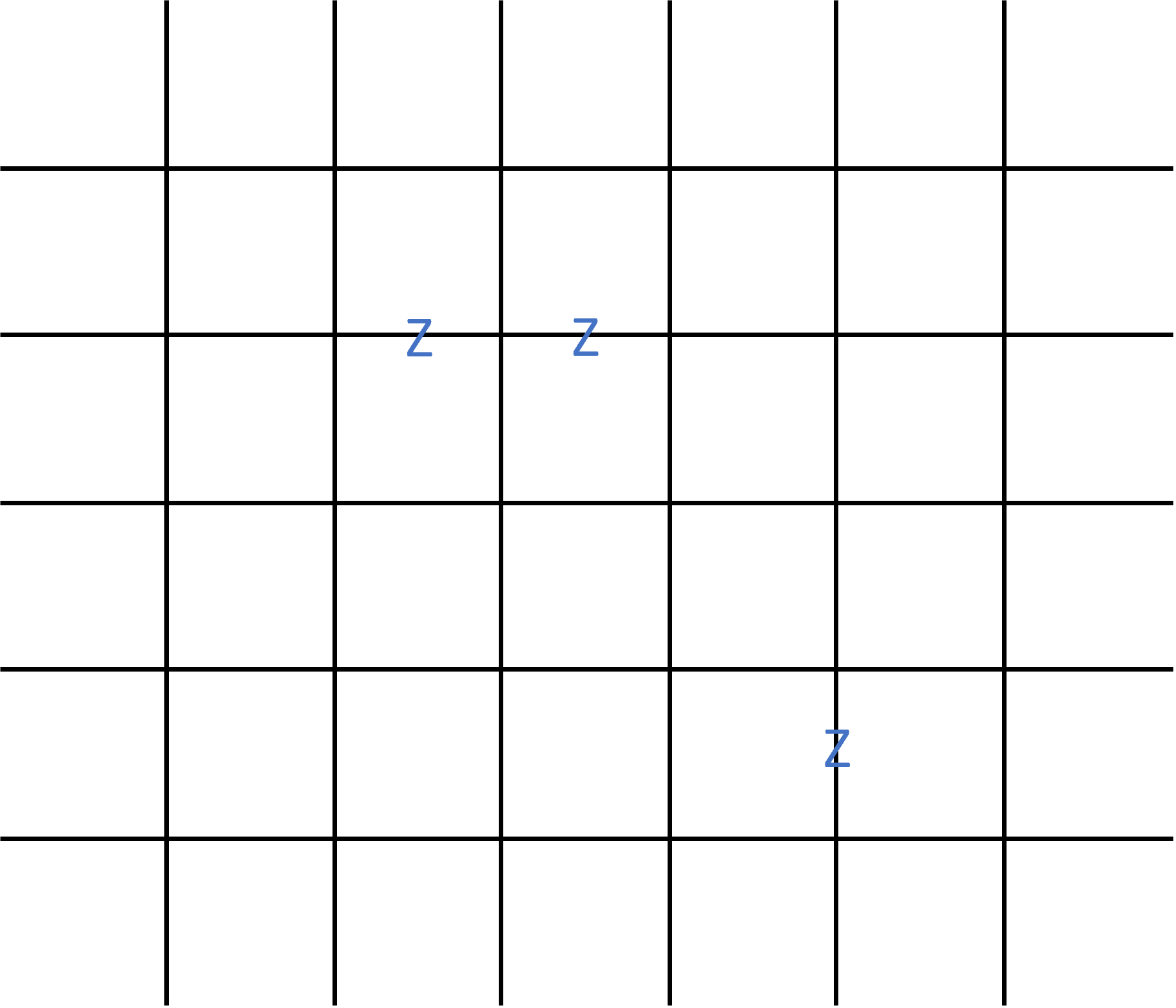} 
			\caption{\scriptsize $Z$ correction}
			\label{fig:3.6(h)}
		\end{subfigure}
		\caption[Decoding process of UF decoder]{	 UF decoding for a toric code. (a) The two thick blue edges represent erasures, and three $Z$ errors are marked on the edges. Red squares indicate nontrivial $X$-check nodes. (b)  Four invalid clusters are highlighted with an orange halo. (c)  All invalid clusters expand outward by one half-edge.  Two invalid clusters in the bottom-right corner are merged into a single valid cluster, now highlighted with a green halo. (d)  The two remaining invalid clusters expand and merge into a valid cluster, while the existing valid cluster remains unchanged. 
			(e) The blue halo represents a spanning forest for the valid clusters. (f) The peeling decoder is applied on the spanning forest to determine the correction operator.}
		\label{fig:3.6}
	\end{figure}
	
	An example of UF decoding  is shown in Figure\,\ref{fig:3.6}.

 	With the use of specific data structures, syndrome validation can be completed in $O(n \alpha(n))$ time in the worst case, where $\alpha(n)$ is the inverse of Ackermann's function, which grows extremely slowly~\cite{Tar75}. 
It has been demonstrated that the worst-case complexity of UF is $O(n)$~\cite{GB24}.

 Moreover, the UF decoder provides a \textit{distance guarantee} for its error-correction capability: it can correct any Pauli errors up to weight  $t=\lfloor\frac{d-1}{2}\rfloor$ on a distance-$d$ toric or surface code.

 \begin{lemma}[Union-Find decoder \cite{DN21}]\label{lemma:UF}
 	There exists a   decoder capable of correcting $r$  erasure errors  and additional $t$ $Z$-errors  
 	on  an $n$-qubit toric or surface code of distance $d$, provided   $r+2t<d$, with complexity $O(n)$.
 	A Similar result holds for $X$ error and erasure correction.
 \end{lemma}

	\subsection{Phenomenological  Noise Model}
 
	In addition to independent depolarizing Pauli errors on the data qubits,  imperfect measurement operations may cause the outcomes $(-1)^0$ and $(-1)^1$ to flip between each other.
	To handle faulty syndrome measurements, we consider the scenario where syndrome extraction is repeated for $d$ rounds for a code of distance $d$, without employing additional redundant stabilizer measurements~\cite{ALB20}. 
	During each round of syndrome extraction, data qubits undergo independent depolarizing errors, while measurement outcomes may suffer independent bit-flip errors. 
	Following~\cite{KL25}, we define the phenomenological noise decoding problem  as follows.

	\begin{definition} \label{def:pndp}
		\noindent{ [Phenomenological noise decoding problem]}: 
		Given a set of $m$ stabilizers of a stabilizer group $\cS$ and  $r$ rounds of syndrome extraction
		outcomes $\sigma^{(\ell)}\in\{0,1\}^m$ 
		of Pauli errors $E^{(\ell)}\in\{I,X,Y,Z\}^{\otimes n}$ and measurement errors $e^{(\ell)}\in \{0,1\}^m$ for $\ell=1,\dots, r$, the goal of phenomenological noise decoding is to find  $\hat{E}^{(\ell)}\in \{I,X,Y,Z\}^{\otimes n}$ and  $\hat{e}^{(\ell)}\in \{0,1\}^m$ for $\ell=1,\dots, r$
		such that the syndrome of $\hat{E}^{(\ell)}$ is $\hat{e}^{(\ell)}+\hat{e}^{(\ell-1)}+\sigma^{(\ell)}+\sigma^{(\ell-1)}\in\{0,1\}^m $,
		where both $\sigma^{(0)}$ and $\hat{e}^{(0)}$ are  set to $0^m$.

	\end{definition} 
	A logical error occurs if $\prod_{\ell=1}^{r} \hat{E}^{(\ell)} E^{(\ell)}$ is not a stabilizer in $\mathcal{S}$, up to a  global phase.

	Pauli errors on the data qubits accumulate over multiple rounds of syndrome extraction, affecting all subsequent syndrome measurements after the round in which they occur. Consequently, the syndrome difference between two consecutive rounds, $\sigma^{(\ell)}+\sigma^{(\ell-1)}$, is used to estimate the Pauli errors $E^{(\ell)}$ and the measurement errors $e^{(\ell)}$ occurring in the $\ell$-th round.

	To decode the phenomenological noise model for 2D topological codes, we first construct a 3D space-time joint decoding graph by introducing a time axis that represents the sequence of syndrome extractions. In addition to edges connecting check nodes within each round, temporal edges are included to link check nodes between consecutive rounds, capturing the effect of measurement errors.

	Figure~\ref{fig:decoding_DS_surface} illustrates the joint decoding graph for the $[[9,1,3]]$ rotated surface code under the phenomenological noise model.
	At each layer, a check node is nontrivial if its syndrome difference with the previous round is nontrivial. Since a measurement error affects the syndrome values in both rounds, it is said to trigger these check nodes. Between two layers, corresponding check nodes are connected by dotted lines, representing potential measurement errors.

	The resulting decoding graph  $G_X$ (with $d$ layers connected by red dashed lines) can be decoded using the UF decoder to correct $Z$ errors. Similarly, $G_Z$ (with $d$ layers connected by blue dashed lines) is used to decode $X$ errors.

	\begin{figure}[htbp] 
		\centering
		\includegraphics[height=!,width=0.4\textwidth, keepaspectratio=true]{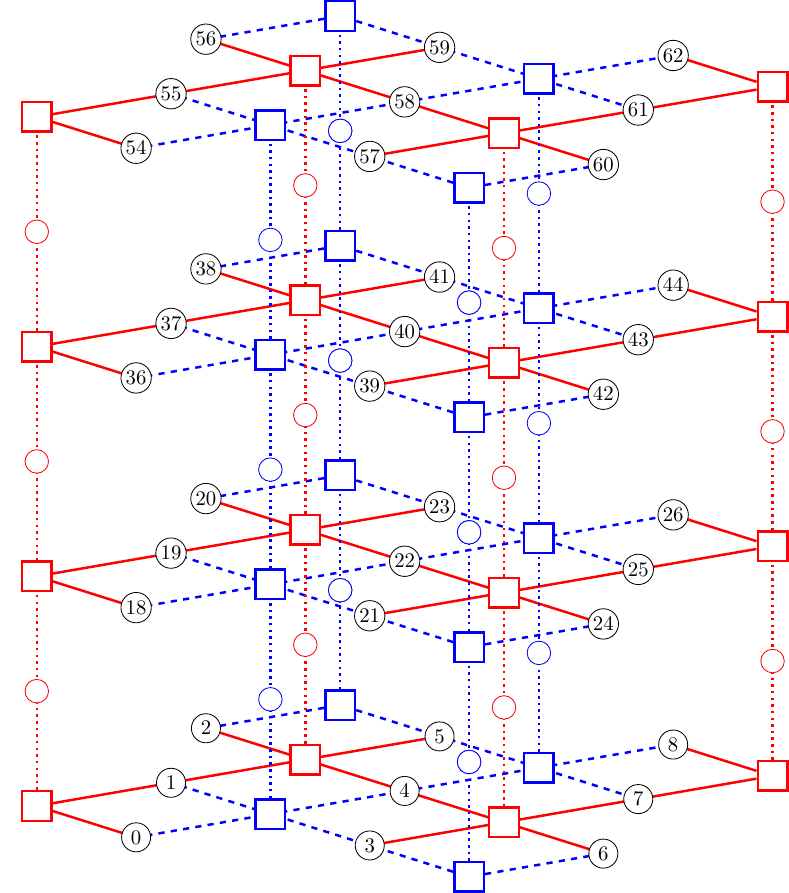}
		\caption{Decoding graphs for the $[[9,1,3]]$ rotated surface code under the phenomenological noise model. Four rounds of syndrome extraction are depicted from bottom to top across four layers.
		Circles with labels represent data qubit errors, while red and blue circles indicate measurement errors for $X$ and $Z$ stabilizer measurements, respectively. Virtual check nodes are omitted for clarity.
The final round is assumed to be free of measurement errors.
	 }
		\label{fig:decoding_DS_surface}
	\end{figure}

	\section{Enhancing the Performance of CSS Codes for Correlated Errors}\label{sec:IRUF}

To evaluate the performance of a   code of distance $d$ under a decoder, we define a weight enumerator of  undecodable errors:
\begin{align}
L(x,y)=  \sum_{i=1}^n a_i  x^i y^{n-i} ,\label{eq:LER}
\end{align}
where $a_i$ represents the number of $n$-fold Pauli operators of weight  $i$  that the decoder fails to correct. Then the logical error rate  of the decoder for the quantum code at  depolarizing  rate $\epsilon$ is $L(\frac{\epsilon}{3}, 1-\epsilon)$ (see, e.g., \cite{LB12}).
If the decoder guarantees correction for any errors of weight  up to  $t=\lfloor \frac{d-1}{2}\rfloor$, then
$a_1=\cdots=a_t=0$. Thus,  undecodable errors of weight $t+1$ are called \textit{critical errors}, as they dominate the logical error rate in the low-error regime. A good decoder should minimize the value of  $a_{t+1}$, while 
ensuring that  $a_1=\cdots=a_t=0$.

	\subsection{$X/Z$ correlations}
For CSS code decoding, treating $X$ and $Z$ errors separately simplifies the problem. However, for error models with correlations, such as depolarizing errors, this approach is suboptimal because it neglects those correlations~\cite{DT14}.  

	Specifically, a $Y$-error is the product of both an $X$- and a $Z$-error. Let $E_j = E_j^X E_j^Z$ represent a Pauli error on qubit $j$ generated by a depolarizing error with rate $\epsilon$, where $E_j^X \in \{I, X\}$ and $E_j^Z \in \{I, Z\}$. It follows that the conditional probabilities are $\mathbb{P}(E_j^X = I \mid E_j^Z = Z) = 1/2$ and $\mathbb{P}(E_j^X = X \mid E_j^Z = Z) = 1/2$. Consequently, when $E_j^Z = Z$ is known, $E_j^X$ can be interpreted as a binary erasure.

	For example,   $X^{\otimes t} \otimes I^{\otimes n-t}$ is decodable by  UF,  but $X^{\otimes t} \otimes Y \otimes I^{\otimes n-t-1}$ is not.
	However, there is only one $Z$-error component in $X^{\otimes t}\otimes Y \otimes I^{\otimes n-t-1}$ that is decodable by $UF$.
 If such a $Z$-error can be precisely located, identifying the position of $Y$ as an erasure simplifies $X$ error decoding.

By leveraging these correlations, $Y$ errors can be handled more effectively. As shown in~\cite{Fow13,DT14,PF23,YL22,iOMFC23},   MWPM   benefits from reweighting the decoding graph using $X/Z$ correlations. In the following subsection, we extend this approach to  UF.

	\subsection{Iterative Union-Find Decoder}

Our first improvement   is called the iterative union-find (IRUF) decoder.  
 IRUF first incorporates the $Z$ correction output from $Z$ error decoding as erasures, thereby enlarging the erasure set for $X$ error decoding. The $X$ correction output is then added to the erasure set for $Z$ error decoding. This process is repeated until a stopping criterion is met.

 In general, the iterations may not converge to degenerate solutions, similar to the behavior observed in the iterative MWPM decoder~\cite{YL22,iOMFC23}. 
 Thus, our stopping criterion is when the maximum number of iterations, $\mathrm{iter\_{max}}$, is reached. 
 The IRUF decoder is summarized in Algorithm~\ref{alg:IRUF}. 
 The time complexity of  Algorithm \ref{alg:IRUF} is clearly $O(\mathrm{iter\_{max}} \cdot n)$.

	\begin{algorithm}
		\caption{Iterative Union-Find decoder (IRUF)} \label{alg:IRUF}
 \textbf{Input:} decoding graphs $G_X$ and $G_Z$, erasure set $\mathcal{E}$, nontrivial  check nodes $\Sigma_X$ and $\Sigma_Z$,  and the maximum number of iterations $\mathrm{iter\_{max}}$.
 
 \textbf{Output:} $X$ correction $\cC_X$ and $Z$ correction $\cC_Z$.

 \textbf{Steps:}
		\begin{algorithmic}[1]
			\State $\cC_X$ = \Call{Union-Find}{$G_Z,\cE,\Sigma_Z$};
			\For{$i = 1$ to $i = \mathrm{iter\_{max}}$}
			\State $\cC_Z$ = \Call{Union-Find}{$G_X,\cE\cup \cC_X,\Sigma_X$};
			\State $\cC_X$ = \Call{Union-Find}{$G_Z$, $\cE\cup\cC_Z$, $\Sigma_Z$};
			\EndFor
			\State \Return $\cC_X$ and $\cC_Z$.
		\end{algorithmic}
	\end{algorithm}

	\begin{figure}[htbp] 
		\centering \scriptsize
		\begin{subfigure}{0.15\textwidth}
			\centering
			\includegraphics[height=!,width=0.99\columnwidth, keepaspectratio=true]{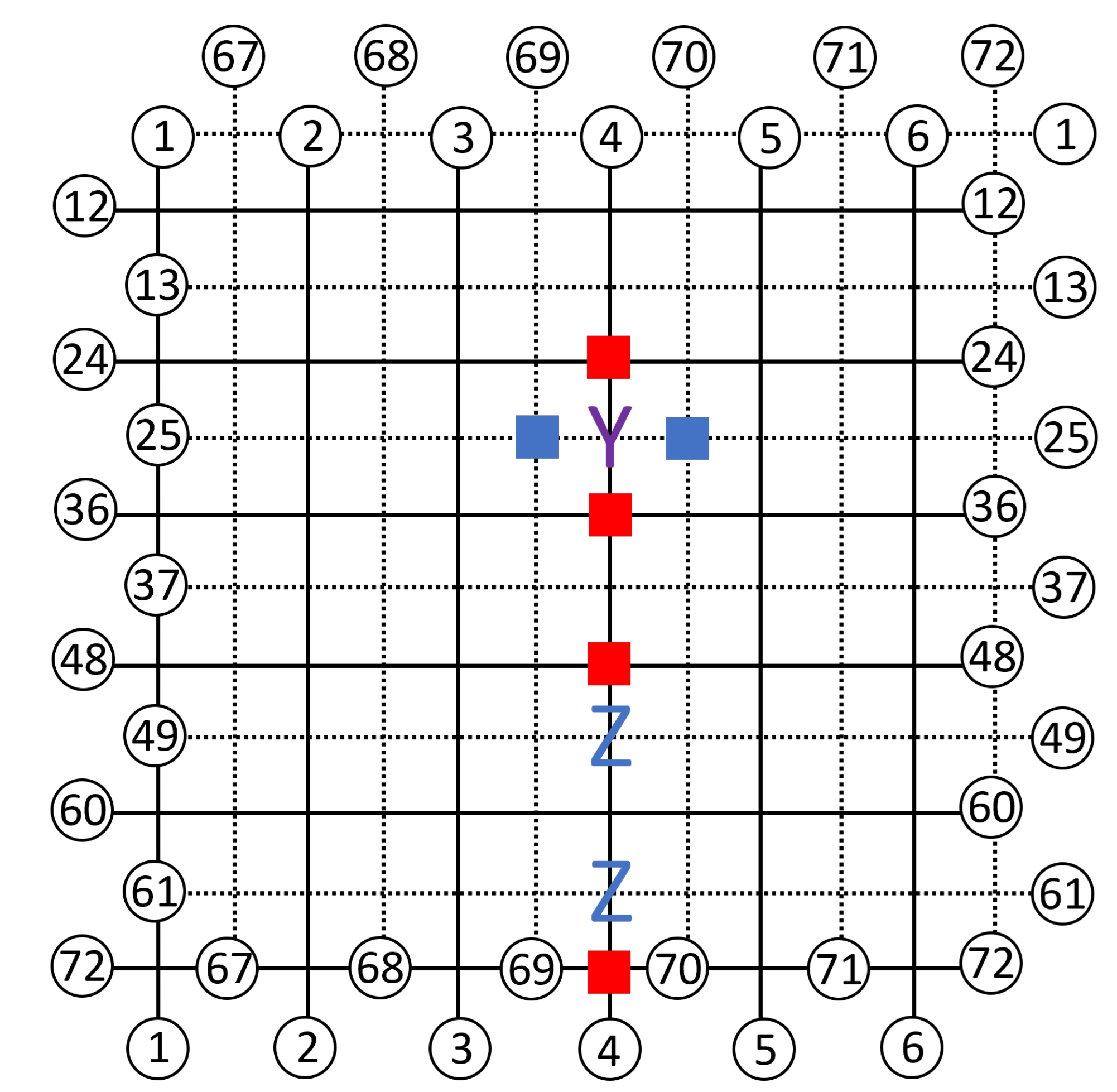} 
			\caption{\scriptsize Errors and  syndromes on   $G_X$ and $G_Z$}
		\end{subfigure}
\hfil
		\begin{subfigure}{ 0.15\textwidth}
			\centering
			\includegraphics[height=!,width=0.99\columnwidth, keepaspectratio=true]{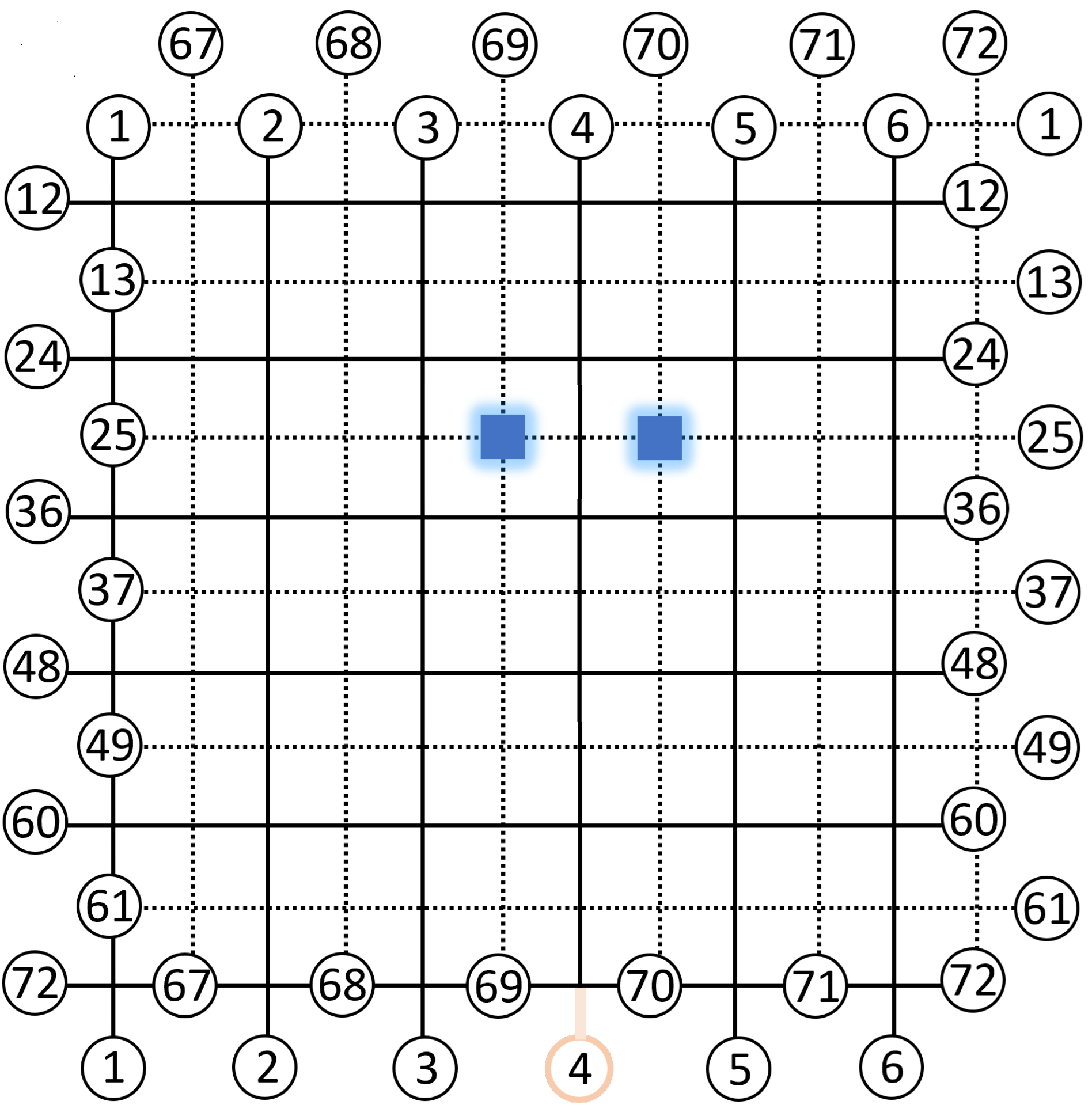} 
			\caption{ \scriptsize Initializing $G_Z$\\ \mbox{  \qquad}}
		\end{subfigure}
\hfil
		\begin{subfigure}{0.15\textwidth}
			\centering
			\includegraphics[height=!,width=0.99\columnwidth, keepaspectratio=true]{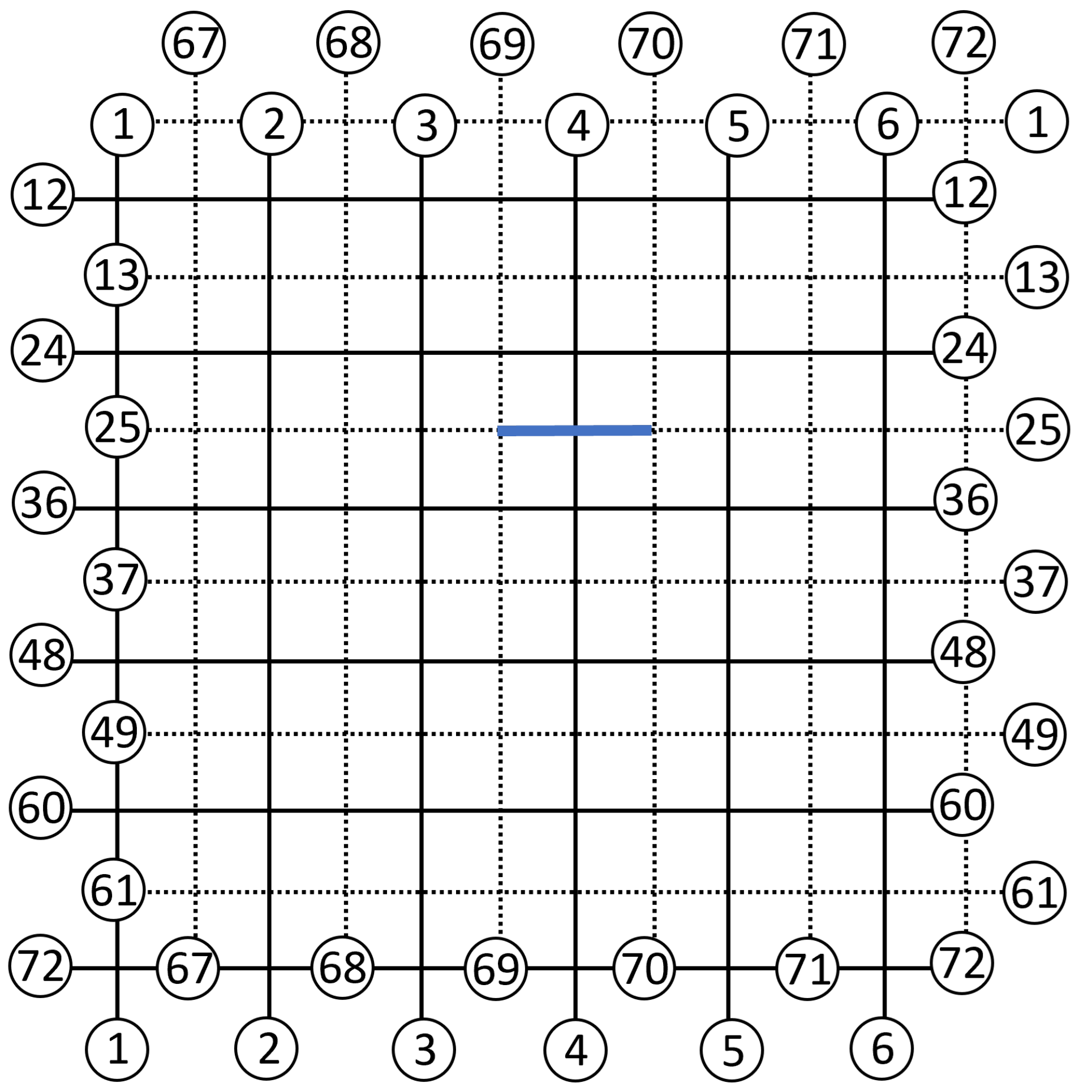} 
			\caption{ \scriptsize Decoding outcome $\cC_X^{(0)}$ }
		\end{subfigure}

					\begin{subfigure}{0.15\textwidth}
					\centering
					\includegraphics[height=!,width=0.99\columnwidth, keepaspectratio=true]{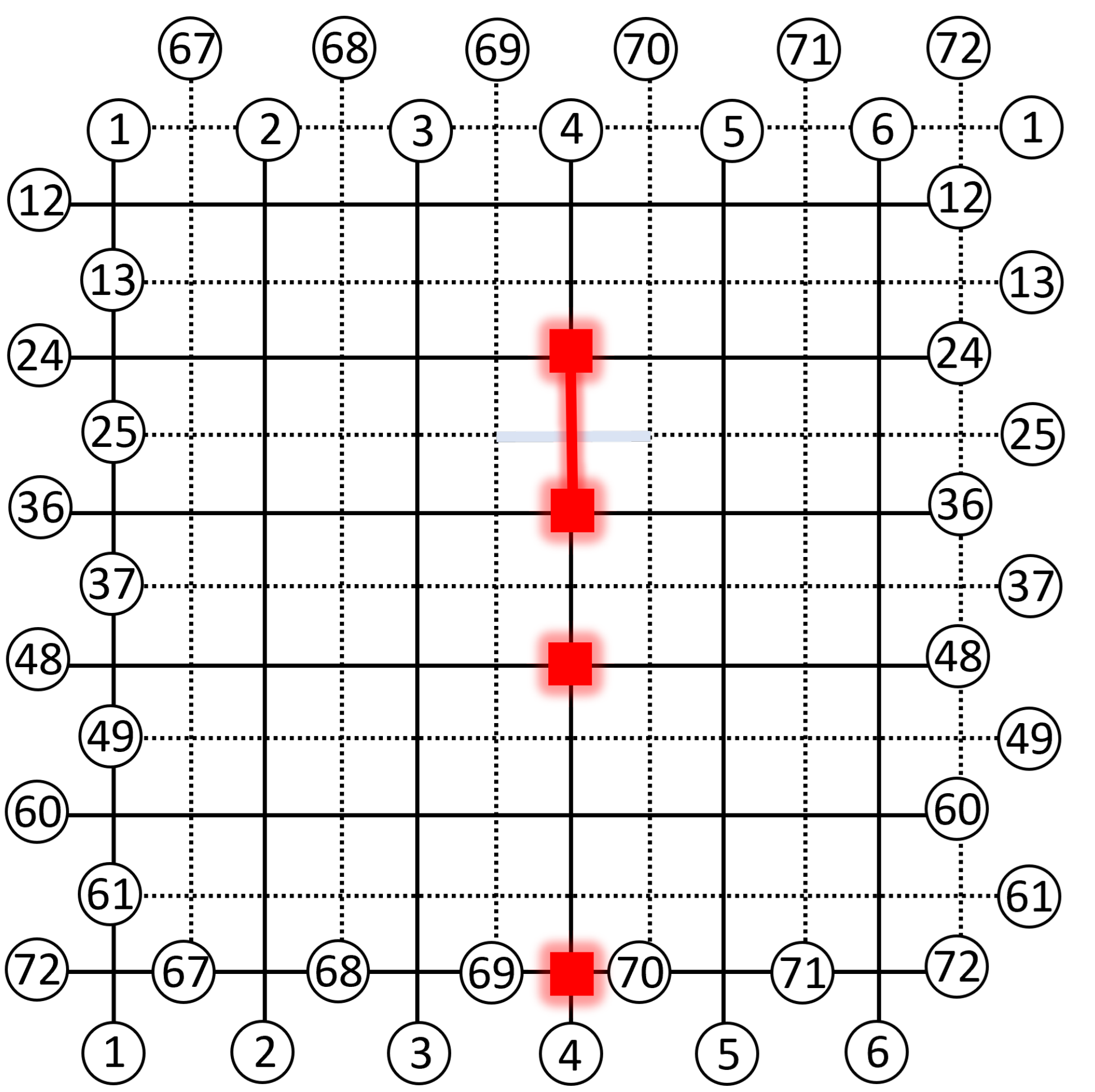} 
					\caption{ \scriptsize Initializing $G_X$ from  $\cC_X^{(0)}$ \\ \mbox{  \qquad}}
				\end{subfigure}
	\hfil
				\begin{subfigure}{0.15\textwidth}
					\centering
					\includegraphics[height=!,width=0.99\columnwidth, keepaspectratio=true]{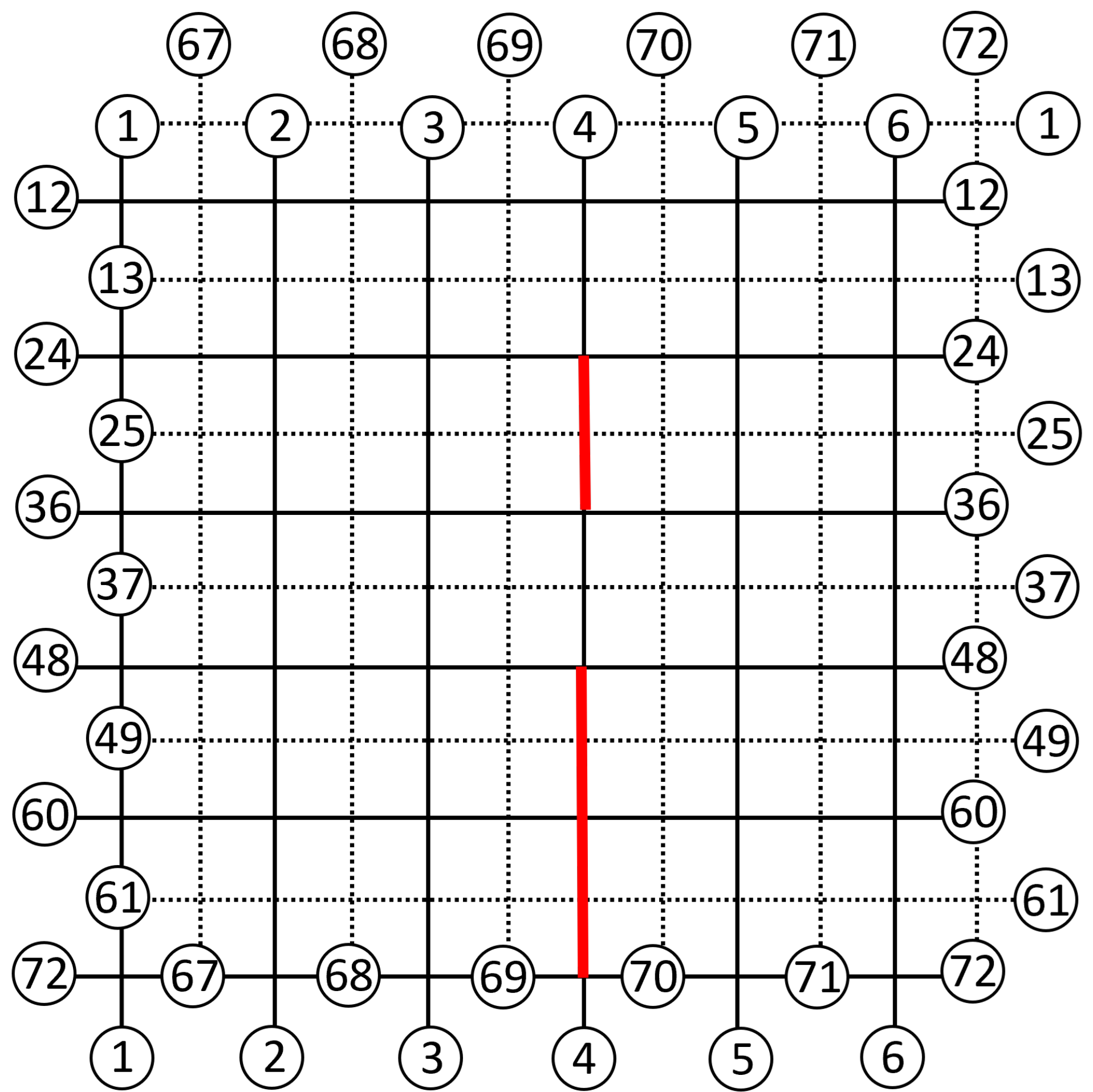} 
					\caption{ \scriptsize Decoding outcome of $\cC_Z^{(1)}$ \\ \mbox{  \qquad}}
				\end{subfigure}
\hfil
		\begin{subfigure}{0.15\textwidth}
	\centering
	\includegraphics[height=!,width=0.99\columnwidth, keepaspectratio=true]{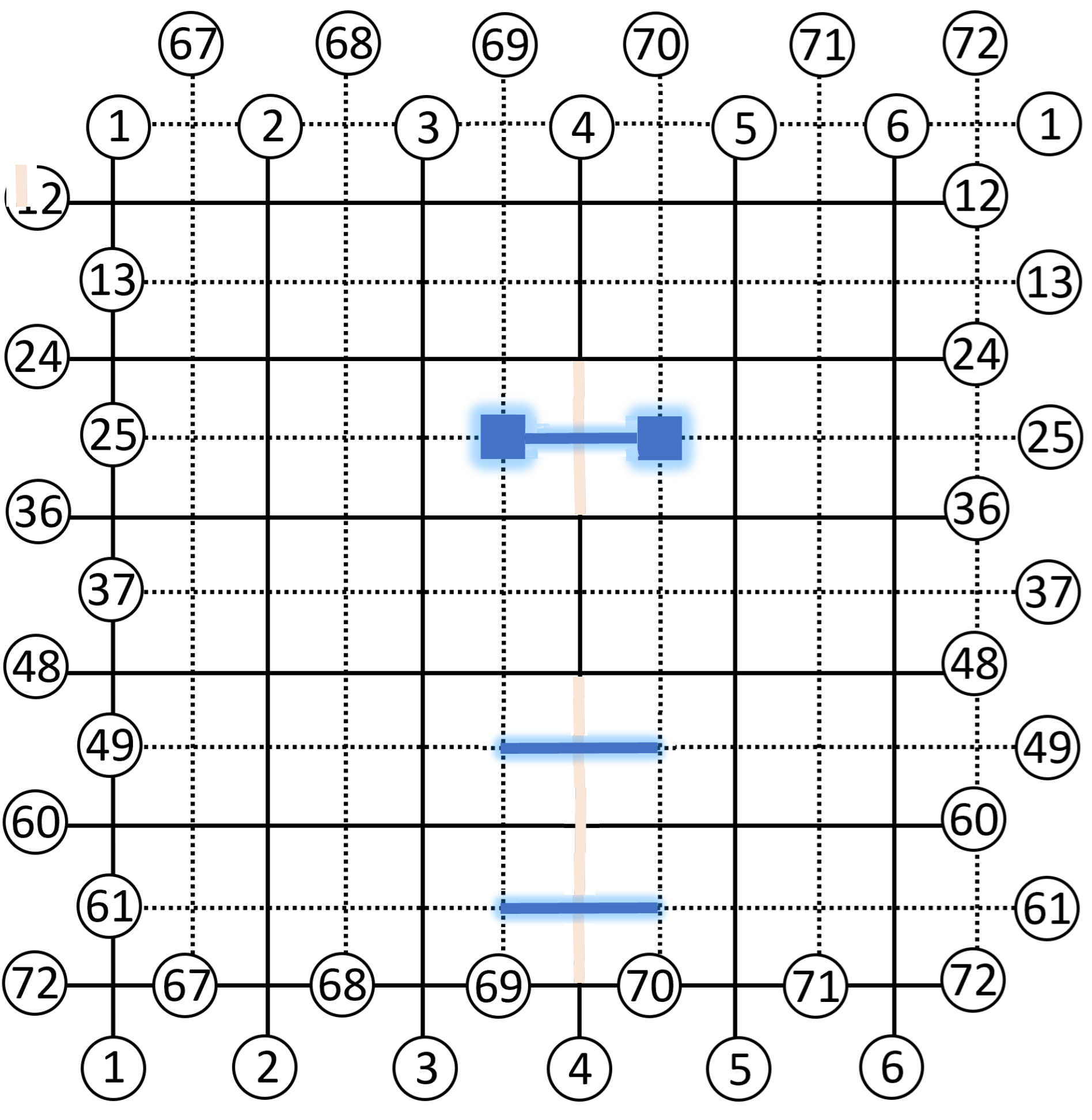} 
	\caption{\scriptsize Initializing $G_Z$ from $\cC_Z^{(1)}$, resulting in   $\cC_X^{(1)}=\cC_X^{(0)}$.}
\end{subfigure}

		\caption[Decoding process of IRUF decoder]{An example of IRUF decoding for a weight-3  error  $Y_{28}Z_{52}Z_{64}$ on the $[[72,2,6]]$ toric code. 
			 	}
		\label{fig:IRUF}
	\end{figure}

		Figure\,\ref{fig:IRUF} illustrates an example of IRUF decoding on the $[[72,2,6]]$ toric code, capable of correcting errors with weight up to two.
		Consider  a weight-3 error  $Y_{28}Z_{52}Z_{64}$. 
		Figure\,\ref{fig:IRUF}~(a) shows the joint decoding graph of $G_X$ and $G_Z$, where red squares represent nontrivial $X$-check nodes and blue squares represent nontrivial $Z$-check nodes. Solid lines denote edges of $G_X$, and dashed lines denote edges of $G_Z$. For clarity, the edges on the borders are labeled with qubit numbers to illustrate the torus structure.
		In this case, the UF decoder  outputs the $X$ correction $X_{28}$   but applies the $Z$ correction  $Z_4 Z_{16} Z_{40}$, resulting a logical $Z$ error. 
		
	While the $[[72,2,6]]$ code cannot guarantee the correction of weight-3 errors due to its minimum distance of 6, the IRUF decoder is possibly able to handle such error events by leveraging error correlations.
	In the first step,  IRUF invokes the UF decoder to address the initial $X$-errors using Figure\,\ref{fig:IRUF}~(b), producing a potential outcome $\cC_X$ as shown in Figure\,\ref{fig:IRUF}~(c), which is correct. With the additional erasure indicated in Figure\,\ref{fig:IRUF}~(d),  the subsequent decoding cycle for $Z$ errors produces the correct $Z$ correction $\cC_Z$ as shown in Figure\,\ref{fig:IRUF}~(e), which is then used to initialize  the next decoding graph $G_Z$. The additional erasures  shown in Figure\,\ref{fig:IRUF}~(f) do not affect the $X$ correction and the UF decoder again outputs the correct $X$ correction, as illustrated in Figure\,\ref{fig:IRUF}~(c).

 \begin{table}[htbp]
 	\centering
 	\scriptsize
 	\begin{tabular}{|c|c|c|c|c|c|c|}
 		\hline
 		\multirow{2}{*}{code family} &\multirow{2}{*}{$d$} & weight& \multicolumn{3}{c|}{ undecodable errors} & total     \\
 		 \cline{4-6} 
 		& & $ t+1$& 1 iter &2 iter&10 iter& errors\\
       \hline
 		rotated toric& $10$&5 &46&&& 317619225\\
 		\hline
 		surface& $7$&4 & 8&4&5&2666790 \\
 		\hline
 		rotated surface& $5$&3 &18&15&15&2700  \\
 		\hline
 	\end{tabular}
 	\caption{Smallest topological codes without a distance guarantee using IRUF for different maximum iteration limits. The number of  undecodable errors and the total number of weight-$(t+1)$ errors are provided.
 	}\label{tb:distance}
 \end{table}

 Unfortunately, unlike the distance guarantee of  UF described in Lemma~\ref{lemma:UF}, the IRUF decoder does not provide  a distance guarantee for errors up to weight  $t$.
Table\,\ref{tb:distance} summarizes the number of undecodable errors of weight  $t=\lfloor\frac{d-1}{2}\rfloor$ using IRUF for the smallest topological codes without a distance guarantee.
 Notably, no undecodable errors of weight $\lfloor\frac{d-1}{2}\rfloor$ were found for $[[2d^2, 2, d]]$ toric codes with $d \leq 10$. However, verifying the case of $d = 11$ would require a computational complexity of $1.612 \times 10^{12}$, which exceeds our available resources. Given the geometric similarity between toric codes and the other three code families, we conjecture that IRUF does not provide a distance guarantee for (rotated) toric and surface codes with sufficiently large distances.

	\section{Union-Intersection Union-Find Decoder}
	\label{sec:UIUF}

	In this section, we introduce a non-iterative decoder that achieves improved decoding performance, exhibits better error threshold behavior, and corrects  errors up to weight $t$. This approach incorporates a union-intersection procedure to more effectively leverage $X/Z$ correlations within the UF approach. As a result, this decoder is termed the \textit{union-intersection union-find} (UIUF) decoder.

\subsection{UIUF procedure}
In UF decoding, the first step is syndrome validation, where valid clusters are generated in the decoding graphs $G_X$ and $G_Z$ to ensure that potential $Z$ and $X$ corrections, consistent with the given error syndromes, are contained within these clusters. The peeling decoder is then applied to determine a pair of $Z$ and $X$ corrections.
 	
To utilize $X/Z$ correlations, the $X$ and $Z$ UF decoding outputs are used as erasures for each other in IRUF, and this process iterates to achieve refined results. However, according to the data processing inequality~\cite{CT91}, 
it is more effective to work directly with the valid clusters generated by syndrome validation to utilize $X/Z$ correlations, rather than relying on transformed outputs from the peeling decoder.
 Therefore, to accurately identify the locations of  $Y$ errors,   we focus on valid clusters within both $G_X$ and $G_Z$.
Specifically, qubits whose corresponding edges in    $G_X$ and $G_Z$ appear within valid clusters are likely 
to have $Y$ errors and are thus marked as erasures. 
This process is referred to as the \textit{intersection step}.

 	The full procedure of the UIUF decoder is outlined in Algorithm~\ref{alg:UIUF}. First, syndrome validation is performed, followed by the intersection step. The output of the intersection step is then added to the erasure set. Finally, the UF subroutines are executed to determine the $X$ and $Z$ corrections, respectively.
 	
 	\begin{algorithm}
 	\caption{Union-Intersection Union-Find (UIUF)} \label{alg:UIUF}
 	\textbf{Input:} decoding graphs $G_X$ and $G_Z$, erasure set $\mathcal{E}$, nontrivial  check nodes $\Sigma_X$ and $\Sigma_Z$.
 	
 	\textbf{Output:} $X$ correction $\cC_X$ and $Z$ correction $\cC_Z$.
 	
 	\textbf{Steps:}
 	\begin{algorithmic}[1]
 		
 		\State $\cL_{X}$ = \Call{\SV}{$G_X,\Sigma_X$} \Comment{Union}
 		\State $\cL_{Z}$ = \Call{\SV}{$G_Z,\Sigma_Z$}
 		\For{each qubit $j$}\Comment{Intersection}
 		\If{the edges corresponding to qubit $j$ in $G_X$  and $G_Z$  
 			\State are each covered by a cluster in $\cL_{X}$  and $\cL_{Z}$, 
 			\State				respectively,}
 		\State   $\cE=\cE\cup \{j\}$.
 		\EndIf
 		\EndFor
 		
 		\State $\cC_X$ = \Call{Union-Find}{$G_Z,\cE,\Sigma_Z$};  \Comment{Union-Find}
 		\State $\cC_Z$ = \Call{Union-Find}{$G_X$, $\cE$, $\Sigma_X$};
 		\State \Return $\cC_X$ and $\cC_Z$.
 	\end{algorithmic}
 \end{algorithm}
 
The time complexity of Algorithm~\ref{alg:UIUF} is at most twice that of the UF decoder, as it  adds one additional syndrome validation step for both $X$- and $Z$-errors, and the intersection step is much faster than the spanning forest growth. 
 	Importantly,  the UIUF decoder provides a distance guarantee on its error correction capability as follows.

 \begin{theorem} \label{th:4}
For an $n$-qubit toric or surface code of distance $d$, the UIUF decoder in Algorithm~\ref{alg:UIUF} can correct up to $r$ erasure errors, along with an additional Pauli error of weight $t$ acting on the remaining $n - r$ qubits, provided that $r + 2t < d$.
 \end{theorem}

 \begin{proof}
 	
 	Consider $r$ erasures and a Pauli error $E$ of weight $t$ such that $r + 2t < d$. 
 	Let $E$ contain $s$ Pauli $X$'s, $q$ Pauli $Y$'s, and $t - s - q$ Pauli $Z$'s. In the decoding graph $G_Z$, there are $s+q$ $X$-errors and $r$ erasures; in $G_X$, there are $t - s$ $Z$-errors and $r$ erasures.

To analyze error clustering in the UF decoder, consider the special case  where $E$ consists entirely of $t$ $Z$-errors. Each individual $Z$ error typically results in a pair of adjacent nontrivial check nodes in the decoding graph or a connection to a virtual boundary node. When multiple $Z$-errors are present, the corresponding check nodes merge into larger connected clusters. At each union step, at least one additional half-edge of $E$ is covered.   This implies that the largest possible cluster diameter in a decoding graph is at most $2t$~\cite{DN21}. When $r$ additional erasures are present, they do not increase the number of union steps but can contribute to a larger cluster size. In the worst case, erasures extend the cluster further, leading to a maximum cluster diameter of $2t+r$.

By this reasoning, after the initial union steps, the largest cluster in $\cL_Z$ has a diameter of at most $2s + 2q + r$, while the largest cluster in $\cL_X$ has a diameter of at most $2t - 2s + r$. The intersection step produces an erasure set of diameter  at most $\min\{2s+2q, 2t - 2s\} + r$  for the following UF decoding.

Note that the $Y$ errors are now covered by erasures. Consequently, in $G_Z$, there are $s$ errors and an erasure set with a diameter of at most $\min\{2s+2q, 2t - 2s\} + r$, while in $G_X$, there are $t - s-q$ $Z$-errors and an erasure set with a diameter of at most $\min\{2s+2q, 2t - 2s\} + r$. Since 
 	\begin{align}
 	&\min\{2s+2q, 2t - 2s\} + r + 2s \notag\\
 	&\leq ( 2t - 2s)+ r + 2s = r + 2t < d 
 	\end{align} 
 	and 
	\begin{align}
 	&\min\{2s+2q, 2t - 2s\} +r + 2(t - s-q)\notag \\
 	 &\leq  2s+2q + r + 2(t-s-q) = r + 2t < d, 
	\end{align} 
the largest  clusters in the UF steps in lines 10 and 11 of Algorithm~\ref{alg:UIUF} 
will not cover a nontrivial logical error.  Thus UIUF correctly estimates the $X$ and $Z$ corrections by Lemma~\ref{lemma:UF}.
 	
 	Algorithm~\ref{alg:UIUF} includes two syndrome validations, one intersection step, and two UF subroutine calls. The intersection step, which checks each qubit in $\cL_{X}$ and $\cL_{Z}$, runs in $O(n)$ time. Therefore, the overall complexity of the UIUF decoder remains $O(n)$, consistent with the UF decoder.

 \end{proof}

Consider again the example in Figure \ref{fig:IRUF}. The  initial syndrome validation results on both $G_X$ and $G_Z$ are shown in Figure\,\ref{fig:3.11}~(a).
Here,  the green cross marks the intersection of two edges in  $G_X$ and $G_Z$, which  indicates a shared qubit. Consequently, the  intersection step  identifies   this qubit as an erasure, as shown in Figure\,\ref{fig:3.11}~(b).  The subsequent UF subroutines then yield the correct $X$ and $Z$ corrections.

\begin{figure}[h] 
	\centering
	\begin{subfigure}{0.23\textwidth}
		\centering
		\includegraphics[height=!,width=0.99\columnwidth, keepaspectratio=true]{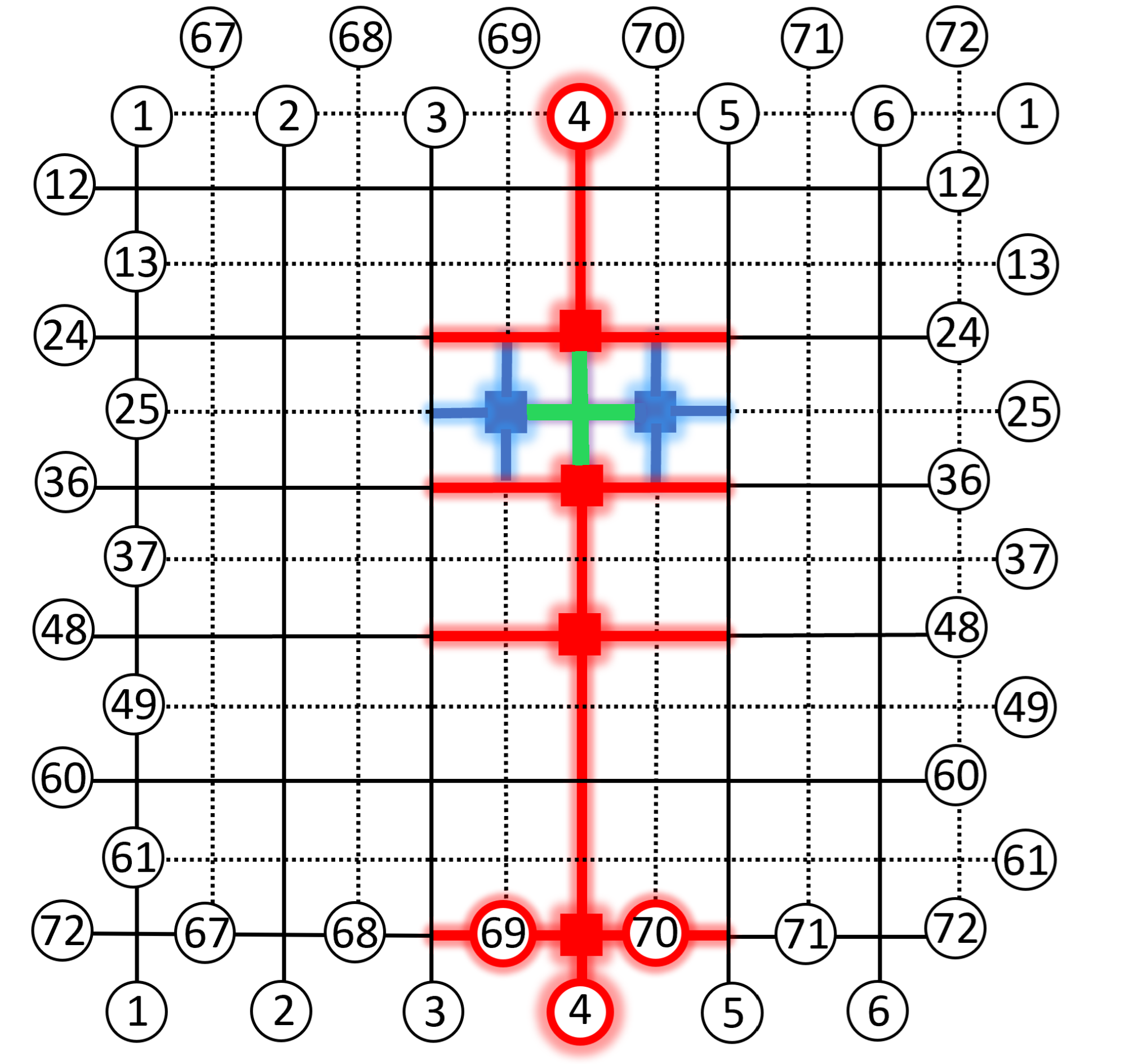} 
		\caption{}
		\label{fig:3.11(a)}
	\end{subfigure}
	\hspace{0.2cm}
	\begin{subfigure}{0.23\textwidth}
		\centering
		\includegraphics[height=!,width=0.99\columnwidth, keepaspectratio=true]{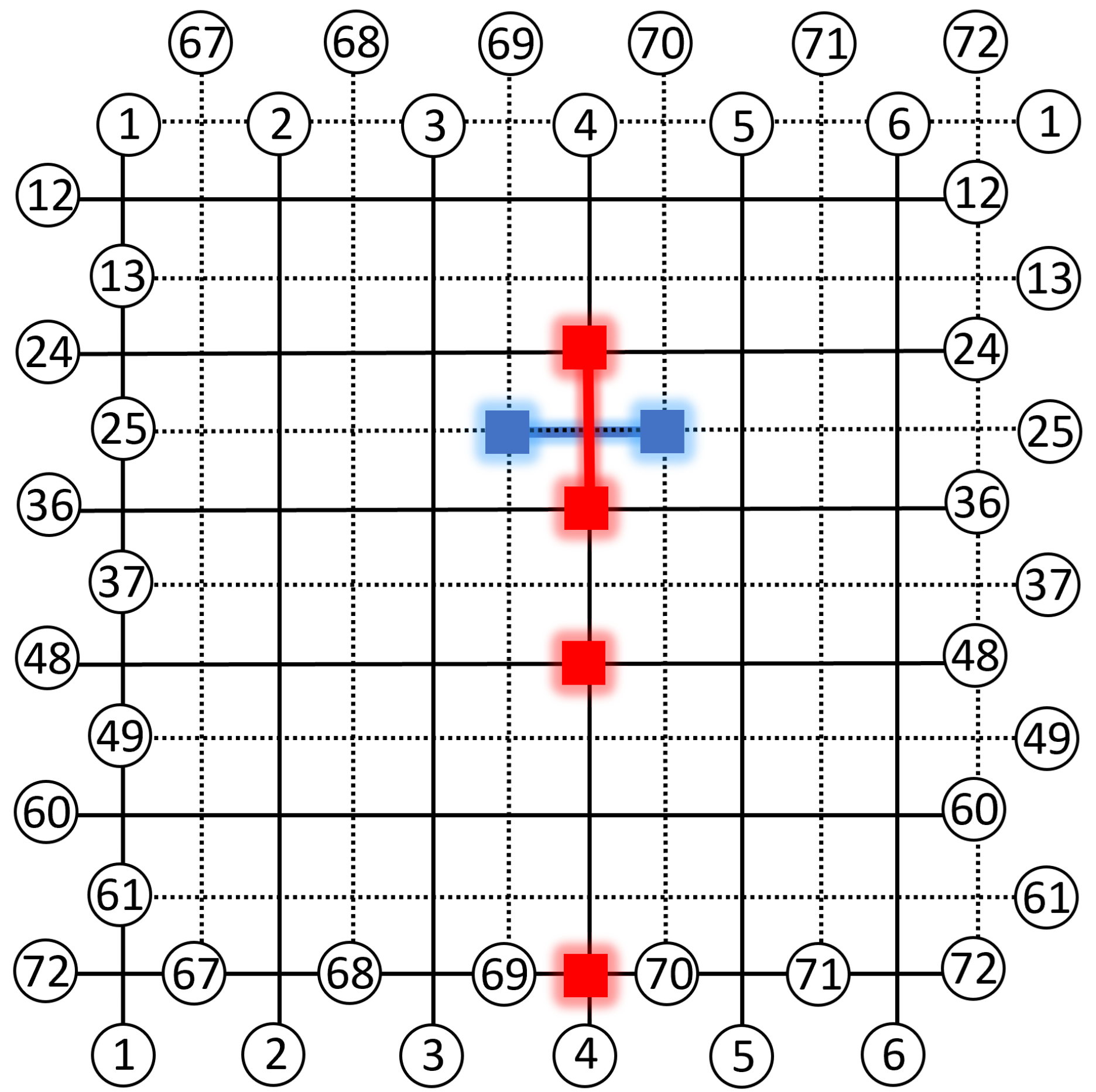} 
		\caption{}
		\label{fig:3.11(b)}
	\end{subfigure}
	\caption[Syndrome validation and intersection step of UIUF decoding]{
		Applying syndrome validation and the intersection step to the example in Figure\,\ref{fig:IRUF}. (a) Results of syndrome validation. (b) Results of the intersection step. }
	\label{fig:3.11}
\end{figure}

\subsection{Critical error analysis}

Consider a decoder with distance guarantee to correct errors of weight up to $t=\lfloor\frac{d-1}{2}\rfloor$.
In general, only Pauli errors with exactly $t+1$ components of either $X$ or $Z$ type are the weight-$t+1$ errors that remain  undecodable  and lack correlations useful for decoding.

	Let us illustrate this with the  $[[36,2,6]]$ toric code with $t=2$.    The UF, IRUF (one iteration), and UIUF decoders all have a distance guarantee on this code.   Table\,\ref{tb:wt3_error} provides the enumeration of errors of different types.
As shown, errors of type $XXX$ and $ZZZ$ that are  undecodable  by UF remain  undecodable  by IRUF and UIUF, as expected. However, IRUF and UIUF can correct significantly more errors of type $YYY$.
	
		\begin{table}[htbp]
			\footnotesize
		\begin{center}
			\begin{tabular}{|c|c|c|c|}
				\hline
			 {Error type} & UF & IRUF  &UIUF\\
				\hline
				all &12358&3056&2108\\
				$XXX$&786&786&786\\
				$ZZZ$&786&786&786\\
				$YYY$&1354&477&225\\
				\hline
			\end{tabular}
		\end{center}
		\caption{The number of  undecodable weight-3 errors for the $[[36,2,6]]$ rotated toric codes using the UF, IRUF (one iteration), and UIUF decoders, respectively.  }\label{tb:wt3_error}
	\end{table}

	Figure \ref{fig:4.1} shows the number of  undecodable   errors of weights $3,$ $4,$ and $5$ for the $[[36,2,6]]$ rotated toric codes using the three decoders, respectively.
One can see that IRUF and UIUF correct significantly more low-weight errors than UF, resulting in a lower logical error rate, which is approximately 4–6 times better in the low error rate regime.

\begin{figure}[htbp]
	\centering
	\includegraphics[height=!,width=0.8\columnwidth, keepaspectratio=true]{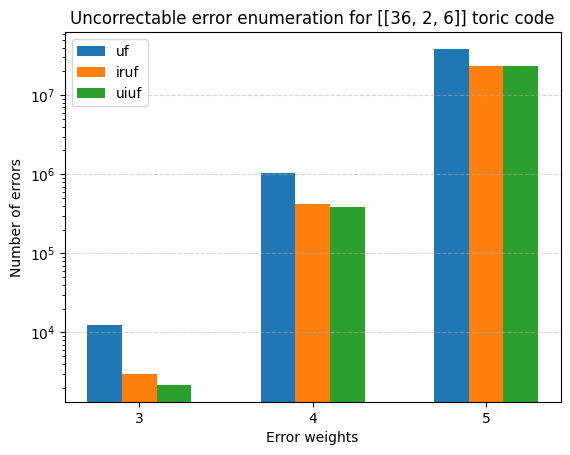} 
	\caption[Enumeration of  undecodable errors  of the three decoders of a toric code]{The number of  undecodable   errors of weights $3,$ $4,$ and $5$ for the $[[36,2,6]]$ rotated toric codes using the UF, IRUF (one iteration), and UIUF decoders, respectively.  }
	\label{fig:4.1}
\end{figure}

	As the error weight increases, the number of undecodable errors in UIUF and UF becomes more similar. This can be explained as follows. All decoders can correct the same total number of errors, determined by the number of distinct syndromes multiplied by the size of the stabilizer group. UF inherently prioritizes decoding low-weight $X$ and $Z$ errors for a given syndrome,  even though their combined effect may result in high-weight errors.
	
	With the intersection step, UIUF can correct more low-weight errors that UF would otherwise misidentify as high-weight errors. However, this advantage comes at the cost of correcting fewer high-weight errors, leading to a shift in the undecodable error distribution. As the error weight increases, there must be a threshold beyond which UF starts to correct more errors than UIUF.

\subsection{Weighted growth}
In \cite{DN21}, a weighted growth (WG) strategy is proposed for syndrome validation, prioritizing the growth of the smallest cluster first. By prioritizing smaller clusters, fewer erasures are introduced overall, potentially increasing the likelihood of success in subsequent erasure decoding steps. This approach naturally extends to UIUF and proves particularly effective, as the accuracy of the intersection step strongly impacts subsequent $X$ and $Z$ decoding. As demonstrated in the following section, our simulations show that WG significantly improves the performance of UIUF.

	\section{Simulation}
	\label{sec:sim}
 In this section, we present numerical simulations comparing the decoding performance of the IRUF and UIUF decoders against the UF and MWPM decoders. These decoders are applied to toric, rotated toric, surface, and rotated surface codes of various distances under depolarizing errors with perfect syndrome measurements
and also the phenomenological noise model with faulty syndrome measurements.

 For the phenomenological noise model, the data qubits experience depolarizing errors with a physical error rate of $\epsilon$, and we assume that the measurement outcomes are flipped with probability $\epsilon$.

 For each performance curve in this section, each data point is obtained by simulating the decoding process until at least $10^4$ failures for logical error rates above $10^{-4}$, 800 failures for logical error rates above $10^{-6}$, and 200 failures for logical error rates below $10^{-6}$.

	\subsection{Decoding performance of UIUF  and IRUF}

	\begin{figure}[h]
		\centering
		\begin{subfigure}{0.49\textwidth}
			\centering
			\includegraphics[height=!,width=0.99\columnwidth, keepaspectratio=true]{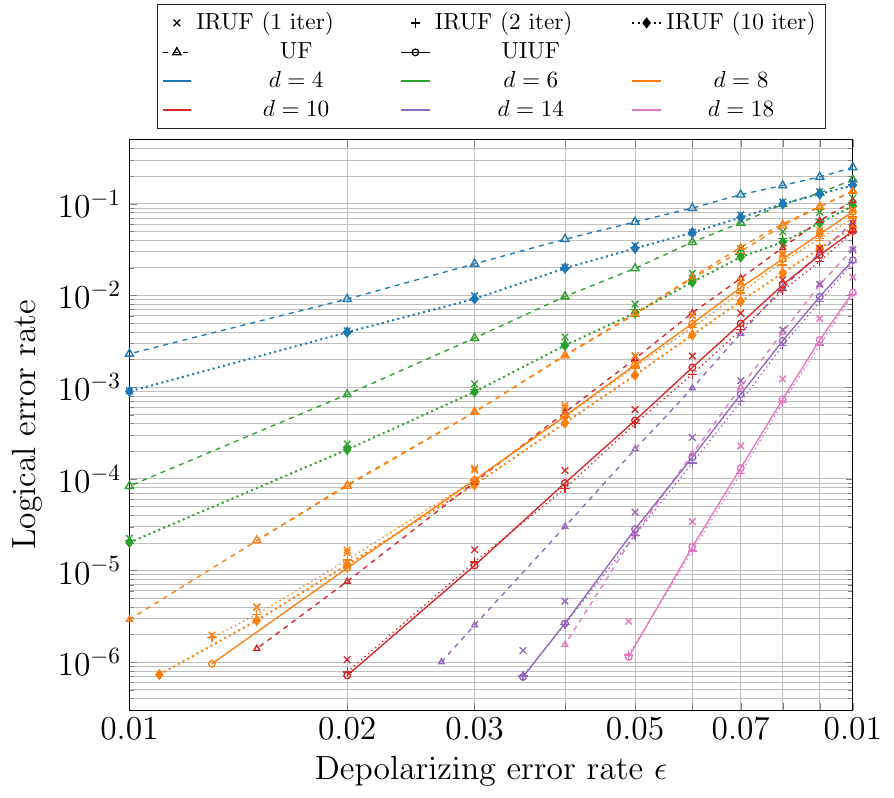} 
		\end{subfigure}
		\caption[Decoding performance comparison of the UF, IRUF, and UIUF decoders of the toric code]{
				Comparison of the decoding performance of   UF, UIUF, and IRUF with $\mathrm{iter_{max}} =1,2$   for $[[2d^2, 2, d]]$ toric codes.
				The performance of IRUF with $\mathrm{iter_{max}} = 10$ is provided for $d = 4, 6, 8$.
		}
		\label{fig:toric_performance}
	\end{figure}
	
	Figure\,\ref{fig:toric_performance} compares the decoding performance of the UIUF and IRUF decoders on  $[[2d^2, 2, d]]$ toric codes with varying maximum iterations, $\mathrm{iter_{max}}$.  Both UIUF and IRUF decoders demonstrate superior performance over the UF decoder.

 Notably, the UIUF decoder reduces the logical error rate by more than an order of magnitude compared to the UF decoder at a logical error rate near $10^{-5}$.
	For example, for $d=14$ at $\epsilon = 0.04$, the logical error rates are $7.1\times 10^{-7}$ for UIUF and $7.6\times 10^{-6}$ for UF.
	Similarly, for $d=10$ at $\epsilon = 0.02$, the logical error rates are $7.1\times 10^{-7}$ for UIUF and $7.6\times 10^{-6}$ for UF.
	Since UIUF guarantees error correction up to the code distance, this performance gap between UIUF and UF persists in the low-error regime.

	IRUF with a single iteration offers a noticeable improvement over UF,  and additional iterations further enhance performance in the high-error regime, particularly for codes with larger distances.
	In contrast, IRUF with  more iterations shows similar performance in the low-error regime. 
	This indicates that additional iterations have a limited impact on reducing the number of critical errors.
It is evident that performance does not improve significantly when $\mathrm{iter_{max}} =10$. 

Additionally, UIUF demonstrates comparable performance to IRUF with two iterations while requiring lower computational complexity.  Moreover, UIUF outperforms IRUF  with  $\mathrm{iter_{max}} =10$ in the low-error regime for $d=8$ due to its distance guarantee, and we expect this trend to hold for other  codes at lower error rates.

	\begin{figure}[htbp]
					\centering				
		\begin{subfigure}{0.48\textwidth}
		\includegraphics[height=!,width=0.99\columnwidth, keepaspectratio=true]{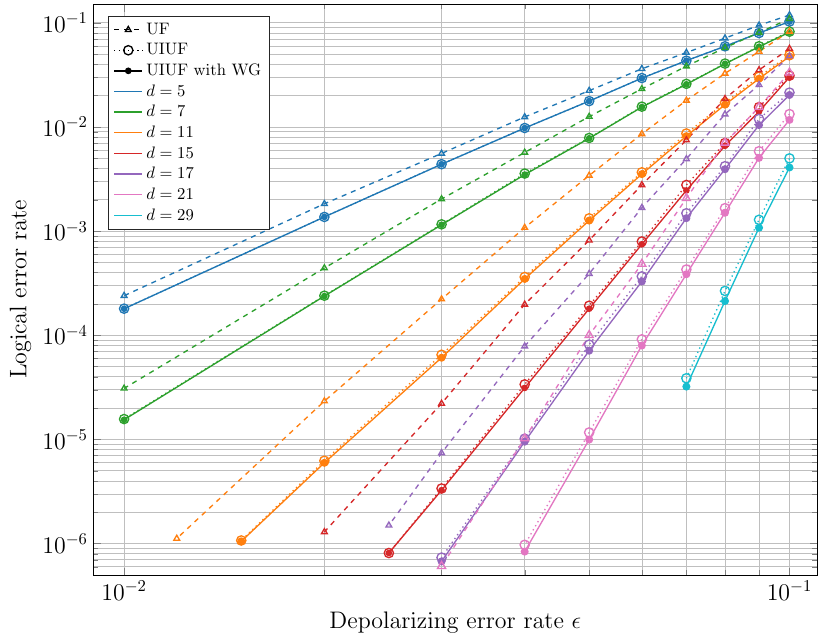} 
		\end{subfigure}
		\caption{	Comparison of the decoding performance of the  UF, UIUF, and UIUF with WG decoders on $[[d^2, 1, d]]$ rotated surface codes.
	}
		\label{fig:UIUF_WG}
	\end{figure}

Similar performance trends are observed for rotated toric codes, surface codes, and rotated surface codes. 
Figure\,\ref{fig:UIUF_WG} shows that UIUF outperforms UF on rotated surface codes, with the improvement being more significant for larger code distances.

 The effect of WG is illustrated in Figure\,\ref{fig:UIUF_WG}. For small distances, UIUF performs similarly with and without WG. However, at larger distances, WG provides a slight improvement in logical error rates.
	For $d=29$, UIUF without WG has a logical error rate approximately $1.2$ times higher than with WG.
This effect will be reflected in the threshold analysis in the next subsection.

	All the above observations also hold for the phenomenological noise model, as shown in Figure~\ref{fig:DS}, where $d+1$ rounds of syndrome extraction are performed for a code of distance $d$. In this setup, the final round is assumed to have perfect syndromes, which are used to determine whether a logical error occurred in the previous rounds. Note that the logical error rate here is not normalized by the number of rounds.

At $\epsilon=0.005$, UIUF with WG improves UF with WG by more than an order of magnitude for the unrotated toric code with $d=10$. Similar trends are observed for larger distances.

\begin{figure}[htbp]
			\centering				
	\begin{subfigure}{0.48\textwidth}
		\includegraphics[height=!,width=0.99\columnwidth, keepaspectratio=true]{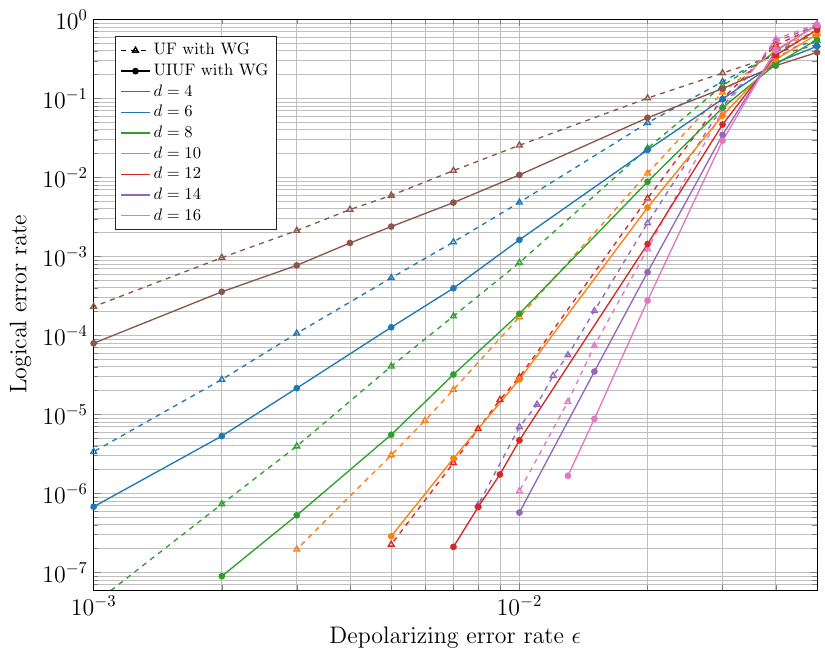} 
	\end{subfigure}
	\caption{	Comparison of the decoding performance of the UF with WG and UIUF with WG decoders for $[[2d^2, 2, d]]$ toric codes under the phenomenological noise model
		with $d+1$ rounds of syndrome extraction.
	}
	\label{fig:DS}
\end{figure}

	\subsection{Threshold}
The threshold value of a decoder for a code family is the critical physical error rate below which logical error rates decrease exponentially with increasing code distance.
In this subsection, we perform a threshold analysis of topological codes using the UF  
and UIUF decoders with WG.

 To determine the threshold values of the topological codes, we use the finite-size critical scaling ansatz~\cite{WHP03}. The logical error rate $p_L$ is fitted to the form $d^{1/\nu} (\epsilon - \tau)$ across a range of topological codes using low-order polynomials, where $\nu$ is the critical exponent and $\epsilon$ is the physical error rate, and $\tau$ is the threshold. Our goal is to find the optimal values for $\nu$, $\tau$, and a low-order polynomial function $f$ that accurately match the equation $p_L = f(d^{1/\nu} (\epsilon - \tau))$. 
 The accuracy of $\nu$ is $\pm 0.01$, and the accuracy of $\tau$ is $\pm 0.01\%$.

 Figure\,\ref{fig:toric_threshold}(a) illustrates the decoding performance of UIUF with WG for the $[[2d^2, 2, d]]$ toric codes near the threshold of this code family, where each data point is generated using $10^6$ samples. Using the finite-size critical scaling ansatz, a threshold of $15.51\%$ is determined, as shown in Figure\,\ref{fig:toric_threshold}(b).  
 A similar estimation is performed for surface codes in Figures~\ref{fig:toric_threshold}(c) and \ref{fig:toric_threshold}(d).
 The determined thresholds are robust that even  using $10^5$ samples for each data point    generates almost identical values in our simulations.

	\begin{figure}[htbp]
				\centering				
	\begin{subfigure}{0.23\textwidth}
		\includegraphics[height=!,width=0.99\columnwidth, keepaspectratio=true]{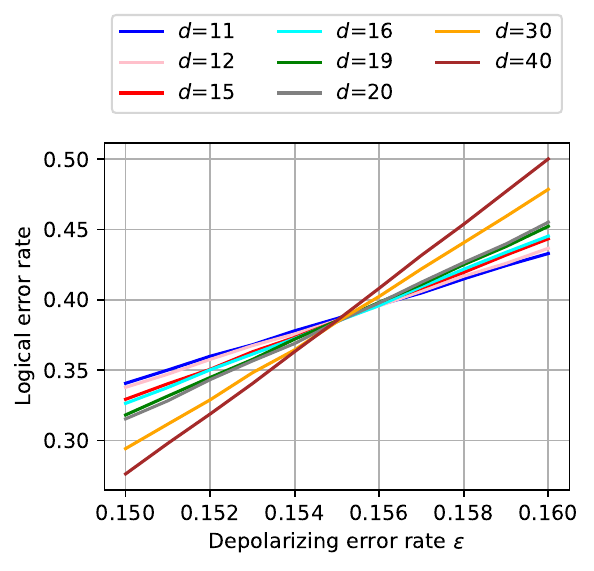} 	
		\caption{Threshold performance of   toric codes.}		
	\end{subfigure}\hspace{0.2cm}
	\begin{subfigure}{0.23\textwidth}
			\includegraphics[height=!,width=0.99\columnwidth, keepaspectratio=true]{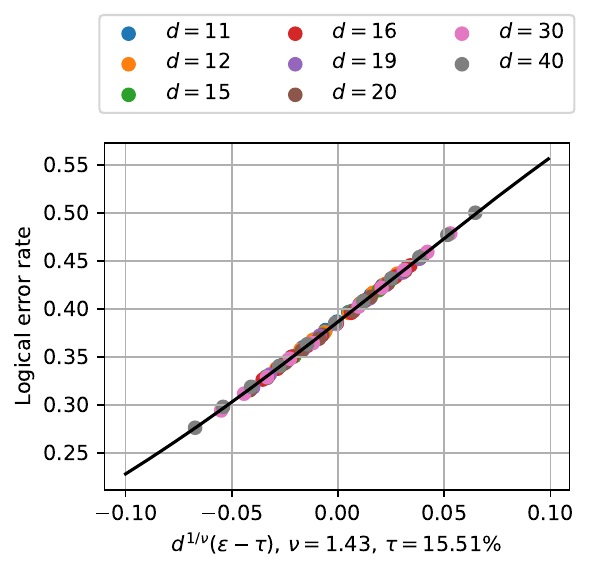} 	
		\caption[Estimated threshold of UIUF decoder of the toric code]{Data fitting using a second-order polynomial.}
	\end{subfigure}\\

	\begin{subfigure}{0.23\textwidth}
	\includegraphics[height=!,width=0.99\columnwidth, keepaspectratio=true]{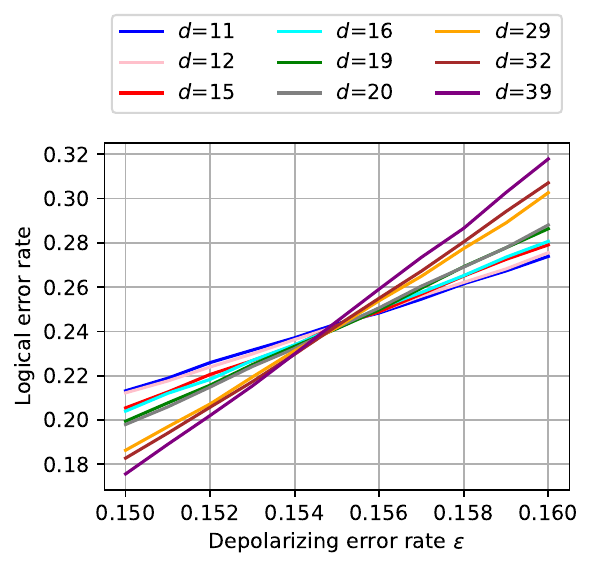} 	
	\caption{Threshold performance of   surface codes.}		
\end{subfigure}\hspace{0.2cm}
\begin{subfigure}{0.23\textwidth}
	\includegraphics[height=!,width=0.99\columnwidth, keepaspectratio=true]{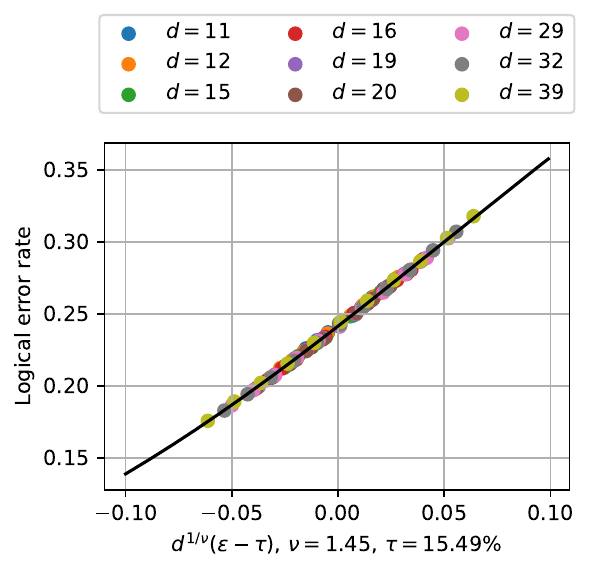} 	
	\caption[Estimated threshold of UIUF decoder of the toric code]{Data fitting using a second-order polynomial.}
\end{subfigure}
	\caption[Threshold of UIUF decoder of the toric code]{
		Estimated thresholds in the code capacity noise model under depolarizing errors using the UIUF decoder with WG: $15.51\%$ for $[[2d^2, 2, d]]$ toric codes and $15.49\%$ for $[[d^2+(d-1)^2, 1, d]]$ surface codes.
	}
	\label{fig:toric_threshold}
 
\end{figure}

\begin{table}[htbp]
	\centering
	\begin{tabular}{|c|c|c|c|c|}
		\hline
		&     UF   &UF  &        UIUF   &  UIUF \\
		&&+ WG&&+ WG\\ \hline
		Toric              &   14.54\%	& 14.93\%  	&     14.72\%   	& 	15.51\% \\ \hline
		Rotated Toric      &   14.44\% 	& 14.87\% 	&    14.72\%    & 	15.52\% \\ \hline
		Surface            &   14.61\%	& 14.92\% 	&    14.98\%  	&	15.49\% \\ \hline
		Rotated Surface    &   14.68\%	& 15.03\% 	&   15.16\% 	&	15.63\%  \\ \hline
	\end{tabular}
	\caption[Table of thresholds of UF, IRUF, and UIUF decoders]{Thresholds of the UF and UIUF decoders with and without WG for the four topological code families.}
	\label{tb:threshold}
\end{table}

The thresholds of the UF and UIUF decoders, with and without WG, for the four topological codes are summarized in Table~\ref{tb:threshold}.
 The threshold of 14.61\% for depolarizing errors for surface codes using the UF decoder is consistent with the value reported in~\cite{MPT22}.

 	In the phenomenological noise model, we estimate an error threshold of $3.55\%$ for the rotated toric codes using the UIUF decoder with WG, as shown in 
 	Figure~\ref{fig:surface_DS_threshold}.  Using the UF with WG decoder instead results in a threshold of $3.45\%$. For toric codes, the threshold values are $3.52\%$ and $3.59\%$ for UF with WG and UIUF with WG, respectively. While these threshold values are similar, Figure~\ref{fig:DS} clearly illustrates the improvement of UIUF over UF.

 	If a lower measurement error rate of $\frac{2\epsilon}{3}$ is simulated, the thresholds for rotated toric codes increase to $4.11\%$ for UIUF with WG and $3.85\%$ for UF with WG.

\begin{figure}[htbp]

		\centering				

\begin{subfigure}{0.23\textwidth}
	\centering
	\includegraphics[height=!,width=0.99\columnwidth, keepaspectratio=true]{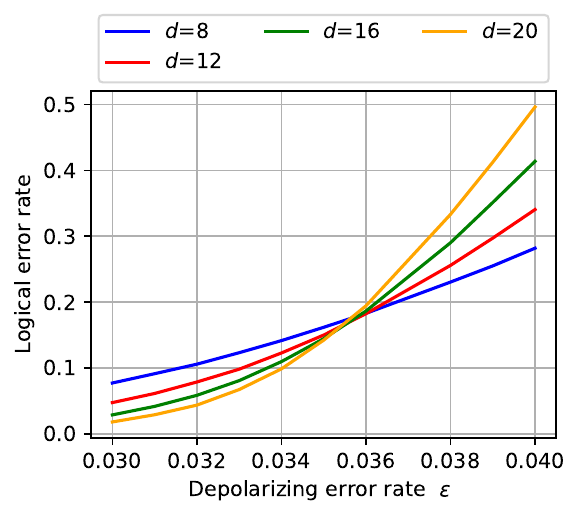} 	
	\caption{Threshold performance for $[[d^2,2,d]]$ rotated toric codes.}		
\end{subfigure}\hspace{0.2cm}
\begin{subfigure}{0.23\textwidth}
	\centering
	\includegraphics[height=!,width=0.99\columnwidth, keepaspectratio=true]{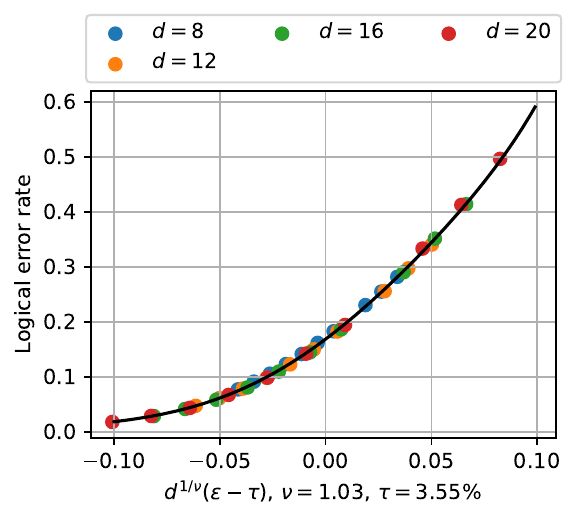} 	
	\caption[Estimated threshold of UIUF decoder of the toric code]{Data fitting using a fourth-order polynomial.}
\end{subfigure}

\caption {
	Estimated threshold in the phenomenological noise model using the UIUF decoder with WG for $[[d^2, 2, d]]$ rotated toric codes.
}
\label{fig:surface_DS_threshold}

\end{figure}

\subsection{Comparison of UIUF and MWPM} 

In this subsection, we compare the UIUF and MWPM decoders. For the MWPM decoder, we use the PyMatching v2 implementation~\cite{HG25} in our simulations. Table~\ref{tb:MWPM_cmp} summarizes the threshold values and time complexity of these decoders.

\begin{table}[htbp]
	\footnotesize
	\begin{center}
		\begin{tabular}{|c| c|c|c|c|}
			\hline
			\multirow{2}{*}{noise model}&\multicolumn{3}{c|}{code capacity} &Phen. \\
			\cline{2-5}
			&               Comp. & toric codes & \multicolumn{2}{c|}{rotated surface codes}\\   
			\hline
			 MWPM& $O(n^3)$ &15.5\%~\cite{WFSH10}  &  15.19\%$^\dagger$&  3.92\% $^\dagger$\\
			  UF+WG   & $O(n)$&  14.93\% &15.03\%&  3.40\%$^\ddagger$ \\
			 UIUF+WG     & $O(n)$&  15.51\% & 15.63\%&3.47\%$^\ddagger$\\
			\hline
		\end{tabular}
	\end{center}
	\caption {Comparison between the MWPM and UIUF decoders in the code capacity and phenomenological noise models.}
	\label{tb:MWPM_cmp}
$^\dagger$: Estimate from PyMatching simulations of logical $X$ error rates.\\
 $^\ddagger$: Estimate using logical $X$ error rates. 
\end{table}

UIUF has a provable worst-case complexity of $O(n)$. While MWPM theoretically has a worst-case complexity of  $O(n^3)$, PyMatching v2 incorporates advanced optimizations that enable near-linear runtime in practical scenarios.

Table~\ref{tb:MWPM_cmp} shows that MWPM and UIUF with WG achieve nearly identical threshold values for toric and surface codes in the code capacity noise model. 
However, for rotated surface codes, UIUF with WG attains a higher threshold of 15.63\%. 
We note that the estimate using logical $X$ error rates is very close to the estimate obtained from overall logical error rates with UF-type decoders.

We remark that the threshold value serves only as a reference for evaluating a decoder's performance. For practical use, decoding performance should be compared, especially in the low-error-rate regime.

 Figure~\ref{fig:MWPM} illustrates   the logical $X$ error rate curves for MWPM and UIUF with WG on rotated surface codes. As seen, MWPM performs better for small distances $3$ and $5$, while UIUF with WG outperforms MWPM for $d \geq 7$, particularly at low-error regime.
 In fact,  UIUF without WG also surpasses MWPM for $d \geq 7$, as can be seen in Figure~\ref{fig:UIUF_WG},  despite having a slightly lower threshold of $15.16\%$ compared to MWPM's $15.19\%$.

Additionally, UF decoding performance can be further improved through preprocessing with neural networks, achieving a threshold of 16.2\% on surface codes~\cite{MPT22}. For comparison, iterative MWPM achieves thresholds of 16.5\%–17\% on (rotated) surface codes~\cite{YL22, iOMFC23}.

	\begin{figure}[htbp]
 		\centering				
				\begin{subfigure}{0.48\textwidth}
					\centering	
						\includegraphics[height=!,width=0.99\columnwidth, keepaspectratio=true]{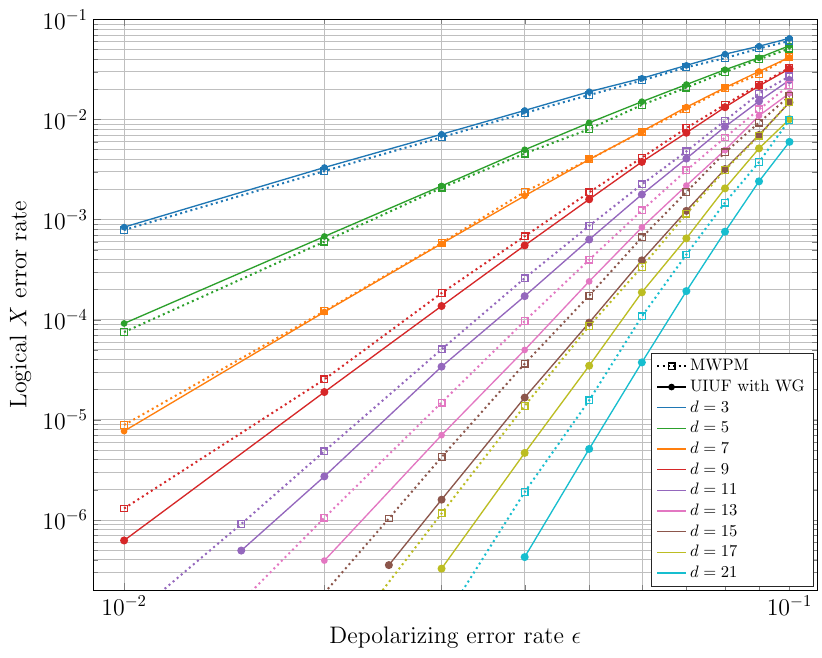} 
					\end{subfigure}
	\caption{
	Comparison of the decoding performance between MWPM and UIUF   for $[[d^2, 1, d]]$ rotated surface codes in the code capacity noise model.
	}
	\label{fig:MWPM}
\end{figure}

In the phenomenological noise model, MWPM exhibits a significantly higher threshold than UIUF. Nonetheless, Figure~\ref{fig:MWPM_DS} compares UF, MWPM, and UIUF, focusing on the logical $X$ error rate. As expected, MWPM outperforms UF for $d \geq 5$. However, UIUF consistently outperforms MWPM in all cases, particularly in the low-error-rate regime. 
For example, the curves for $d=13$ show that while MWPM performs better for $\epsilon$ above 0.03, UIUF with WG demonstrates a clear advantage over MWPM for $\epsilon$ below 0.02.
This advantage arises  from UIUF's ability to explicitly handle more critical correlated errors weight $t+1=\lfloor\frac{d+1}{2}\rfloor$, which are crucial in the low-error-rate regime and become increasingly prevalent over the $d+1$ rounds of syndrome extraction.

	\begin{figure}[htbp]
	\centering				
	\begin{subfigure}{0.48\textwidth}
		\centering	
		\includegraphics[height=!,width=0.99\columnwidth, keepaspectratio=true]{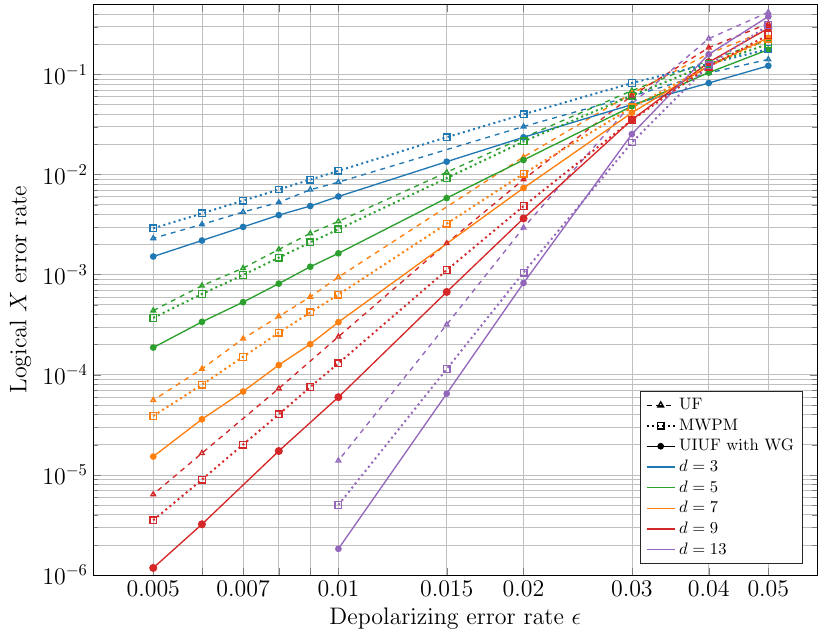} 
	\end{subfigure}
	\caption{
		Comparison of the decoding performance between UF, MWPM, and UIUF   for $[[d^2, 1, d]]$ rotated surface codes in the phenomenological noise model.
	}
	\label{fig:MWPM_DS}
\end{figure}

 	\subsection{Biased noise}
 	
 	While the UIUF decoder is designed for depolarizing errors, we evaluate its performance under varying channel asymmetry in this subsection. 
 	
 	Let $p_X, p_Y, p_Z$ denote the probabilities of single-qubit $X, Y, Z$ errors, respectively. We consider biased noise favoring $Z$ errors, and assume $p_Y = p_X$ for simplicity. The bias ratio is defined as $\eta = \frac{p_Z}{p_X}$, and the total physical error rate is given by $\epsilon = p_X + p_Y + p_Z$.
 	When $\eta = 1$, the noise corresponds to a depolarizing channel. 
 	As $\eta$ increases, the correlation between $X$ and $Z$ errors decreases, leading to a diminishing advantage of UIUF over UF.
 	
 	Figure~\ref{fig:biased} illustrates the performance of the $[[200,2,10]]$ unrotated toric code for various values of $\eta$ at a fixed physical error rate of $\epsilon = 0.05$, using both UF and UIUF decoders with WG. As shown, the performance of UIUF and UF converges as $\eta$ increases, becoming nearly identical for $\eta > 10^3$. 
 	Consequently, UIUF should be used when $\eta < 100$, where it provides a notable improvement. This suggests that UIUF is particularly useful in realistic noise scenarios where errors are not strictly independent but exhibit moderate correlations.

 	\begin{figure}[htbp]
 		\centering
 		\includegraphics[height=!,width=0.45\textwidth, keepaspectratio=true]{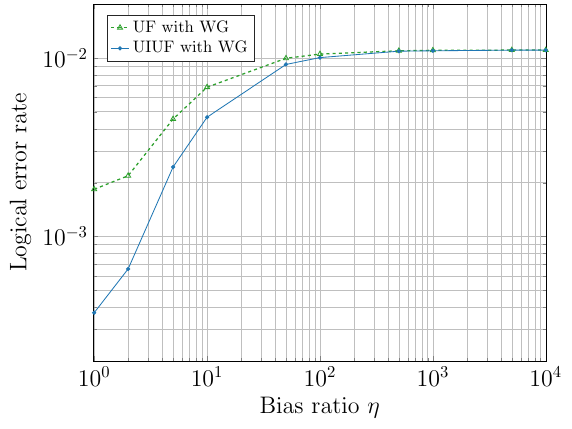} 
 		\caption{Comparison of the UF and UIUF decoders for the $[[200,2,10]]$ unrotated toric code under biased noise with a physical error rate of $\epsilon=0.05$.}
 		\label{fig:biased}
 	\end{figure}

 \subsection{Decoding Time}
 \begin{figure}[htbp]
 	\centering
 	\includegraphics[height=!,width=0.45\textwidth, keepaspectratio=true]{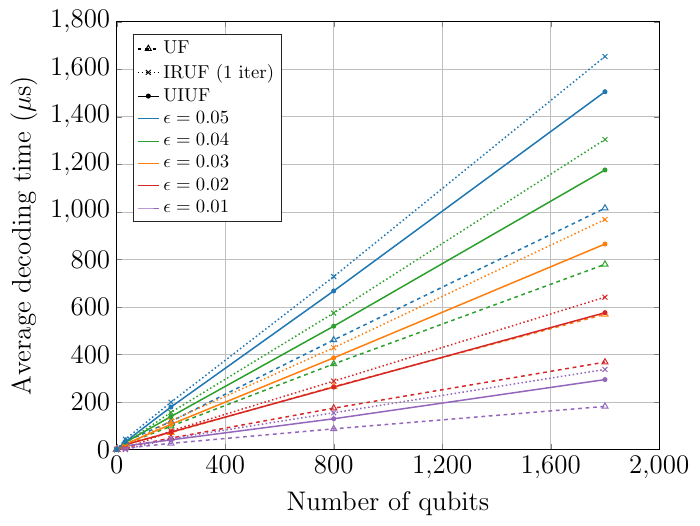} 
 	\caption[Decoding time comparison of the three decoders of the rotated surface code]{Decoding time comparison for the three decoders on the $[[2d^2, 2, d]]$ toric codes with $d = 4, 10, 20, 30$ under the depolarizing noise model. The symbol $\epsilon$ denotes the depolarizing error rate, and the runtime is measured in microseconds.}
 	\label{fig:UIUF_WG0}
 \end{figure}

 The actual average runtime of each decoder on a single-thread computer is shown in Figure \ref{fig:UIUF_WG0}. For each data point, the decoding process was repeated $10^6$ times, and the average decoding time was calculated. We observe that the average runtime of both the IRUF and UIUF decoders scales linearly with the number of qubits, which aligns with the expected time complexity. Specifically, the average runtime of the UIUF decoder is approximately 1.53 times longer than the UF decoder, which is less than twice the time complexity of UF, as anticipated in Algorithm~\ref{alg:UIUF}.

 Furthermore, the UIUF decoder has a shorter average runtime than the IRUF decoder with one iteration, while its performance is comparable to IRUF with more iterations.

 Similarly, we plot the runtime performance of UIUF on rotated surface codes under the phenomenological noise model in Figure~\ref{fig:time_DS}, with $10^6$ trials per data point. In this plot, the $x$-axis represents the number of binary error variables over $d+1$ rounds, given by $d^2(d+1) + (d^2-1)d$. 
The simulated block length is much larger than in the case of perfect syndromes, and thus $10^6$ trials may not be sufficient. Nevertheless, we observe an approximately linear scaling.

  \begin{figure}[htbp]
 	\centering
 	\includegraphics[height=!,width=0.45\textwidth, keepaspectratio=true]{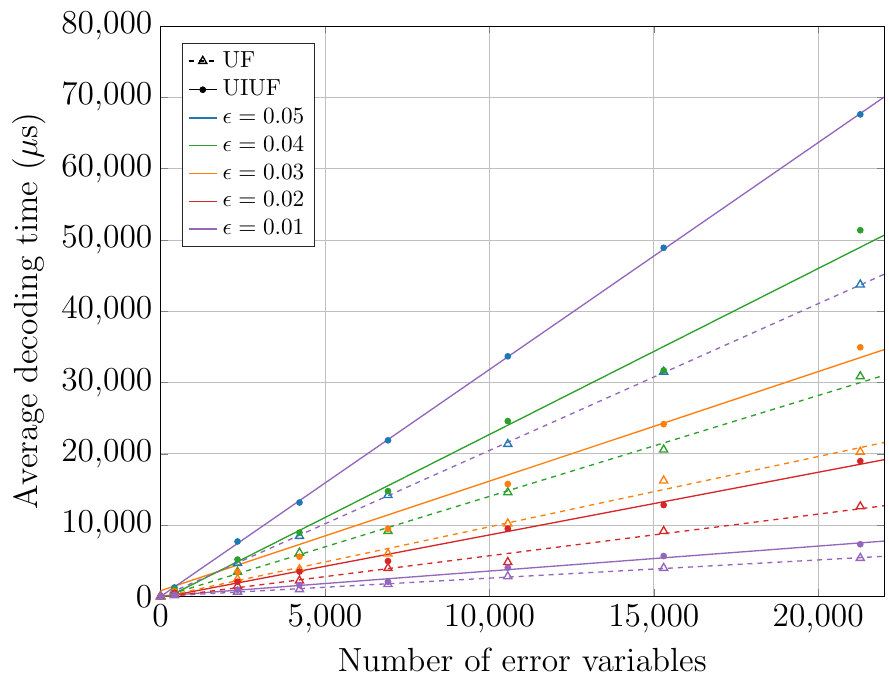} 
 	\caption{Decoding time comparison between the UF and UIUF decoders for $[[d^2, 1, d]]$ rotated surface codes with $d = 5, 9, 11, 13, 15, 17, 19$ under the phenomenological noise model. The symbol $\epsilon$ represents the depolarizing error rate, and the runtime is measured in microseconds.
 		Linear regression lines are drawn for each dataset.
 }
 	\label{fig:time_DS}
 \end{figure}

	\section{Conclusion} \label{sec:con}

In this paper, we have introduced the  UIUF decoder for decoding depolarizing errors in topological codes. By leveraging the correlations between $X$ and $Z$ errors,   UIUF enhances the performance of the standard UF decoder while maintaining its computational efficiency. Unlike iterative decoders such as IRUF, UIUF remains non-iterative with low complexity while ensuring  error correction guarantees.  Our analysis and simulations demonstrate that UIUF outperforms the UF decoder  under both the code capacity and phenomenological noise models, by correcting more low-weight errors through the exploitation of $X$/$Z$ correlations. 
	We have also demonstrated that UIUF outperforms UF under the asymmetric Pauli error model when the bias ratio is sufficiently small.

 Moreover, UIUF outperforms MWPM on rotated surface codes under both the code capacity and phenomenological noise models, making it a promising option for near-term quantum devices.

Future research may focus on optimizing UIUF for more realistic circuit-level errors~\cite{HNB20}. Additionally, exploring hybrid approaches that combine the strengths of UIUF with other advanced decoding techniques~\cite{MPT22} could offer promising avenues for further improving performance.



%

\end{document}